\RequirePackage{amsmath}
\documentclass[runningheads]{llncs}
\usepackage[utf8]{inputenx}
\usepackage[T1]{fontenc}
\usepackage{graphicx}
\usepackage{hyperref}
\usepackage{multicol,multirow}
\usepackage{amsfonts}
\usepackage{amsmath}
\usepackage{algorithm}
\usepackage{xcolor, soul}

\usepackage[noend]{algpseudocode}
\usepackage{tikz}
\usepackage{algpseudocode}
\usepackage{float}
\usetikzlibrary{matrix, shapes, arrows, positioning, chains, calc}
\usetikzlibrary{fit}
\usepackage{ulem}
\usepackage{adjustbox}

\usepackage{enumitem}
\usepackage{bigdelim}
\usepackage{caption} 
\usepackage{subcaption}
\usepackage{bm}
\usepackage{verbatim}
\usepackage{orcidlink}

\usetikzlibrary{shapes, arrows, positioning}

\usepackage[paragraphs, wordspacing, tracking, floats, mathspacing, subtle]{savetrees} 

\newcommand{\ring}[2]{{\mathbb{#1}}_{#2}} 
\newcommand{\field}[2]{\mathbb{#1}_{#2}}
\newcommand{\group}[2]{\mathbb{#1}_{#2}^{*}}
\newcommand{\fieldvec}[3]{\mathbb{#1}_{#2}^{#3}}
\newcommand{\fieldmat}[4]{\mathbb{#1}_{#2}^{#3\times #4}}

\newcommand{\scaler}[1]{#1}
\newcommand{\set}[1]{#1}
\newcommand{\vect}[1]{\bm{#1}}
\newcommand{\vecentry}[2]{\bm{#1}[#2]}
\newcommand{\subvect}[2]{\bm{#1}[\bm{#2}]}
\newcommand{\basisvec}[1]{\bm{e}_{#1}}

\newcommand{\mat}[1]{\bm{#1}}
\newcommand{\matentry}[3]{\bm{#1}[#2,\ #3]}
\newcommand{\matcol}[2]{\bm{#1}[*,\ {#2}]}
\newcommand{\matrow}[2]{\bm{#1}[{#2},\ *]}
\newcommand{\submatcol}[2]{\bm{#1}[*,\ {#2}]}
\newcommand{\submatrow}[2]{\bm{#1}[{#2},\ *]}

\newcommand{\code}[1]{\mathcal{#1}}

\newcommand{\Hatt}[1]{\widehat{#1}}

\newcommand{\Exp}[1]{\mathbb{E}\left[#1\right]}

\begin{document}
\newif\ifsubmission
\submissionfalse
%
\title{ZKFault: Fault attack analysis on zero-knowledge based post-quantum digital signature schemes }
\titlerunning{ZKFault}
%
%
\ifsubmission
\else
\author{Puja Mondal\inst{1}\orcidlink{0009-0006-7300-8435} \and
Supriya Adhikary\inst{1}\orcidlink{0000-0002-0701-8049} \and
Suparna Kundu\inst{2}\orcidlink{0000-0003-4354-852X} \and\ 
Angshuman Karmakar\inst{1,2}\orcidlink{0000-0003-2594-588X}}
\authorrunning{Mondal et al.}
%
\institute{Department of Computer Science and Engineering, IIT Kanpur, India \\ \email{\{pujamondal,adhikarys,angshuman\}@cse.iitk.ac.in}\and
COSIC, KU Leuven, Kasteelpark Arenberg 10, Bus 2452, B-3001 Leuven-Heverlee, Belgium\\
\email{\{suparna.kundu\}@esat.kuleuven.be}
}

\fi

\maketitle              
\begin{abstract}
Computationally hard problems based on coding theory, such as the syndrome decoding problem, have been used for constructing secure cryptographic schemes for a long time. Schemes based on these problems are also assumed to be secure against quantum computers. However, these schemes are often considered impractical for real-world deployment due to large key sizes and inefficient computation time. In the recent call for standardization of additional post-quantum digital signatures by the National Institute of Standards and Technology, several code-based candidates have been proposed, including LESS, CROSS, and MEDS. These schemes are designed on the {relatively new} zero-knowledge framework. 
Although several works analyze the hardness of these schemes, there is hardly any work that examines the security of these schemes in the presence of physical attacks. 

In this work, we analyze these signature schemes from the perspective of fault attacks. All these schemes use a similar tree-based construction to compress the signature size. We attack this component of these schemes. Therefore, our attack is applicable to all of these schemes. In this work, we first analyze the LESS signature scheme and devise our attack. Furthermore, we showed how this attack can be extended to the CROSS signature scheme. Our attacks are built on very simple fault assumptions. Our results show that we can recover the entire secret key of LESS and CROSS using as little as a single fault. Finally, we propose various countermeasures to prevent these kinds of attacks and discuss their efficiency and shortcomings.



\end{abstract}

%
%

\section{Introduction}\label{sec:introduction}
Digital signature schemes are one of the most used and fundamental cryptographic primitives. The security of our current prevalent digital signature schemes based on integer factorization~\cite{RSA} or elliptic curve discrete logarithms~\cite{ECC_miller_Crypto86} can be compromised using a large quantum computer~\cite{Shor_1994,Proos_Zalka_2003}. Therefore, we need quantum computer-resistant digital signature algorithms. 
{In 2022, the National Institute of Standards and Technology (NIST) selects three post-quantum digital signature schemes{~\cite{nist_final_report}} CRYSTALS-DILITHIUM{~\cite{CRYSTALS_Dilithium_Digital_Signatures_from_Module_Lattices}}, FALCON, and SPHINCS+{~\cite{The_SPHINCS_Signature_Framework}} for standardization.} {Among them, FALCON and DILITHIUM are based on lattices, and SPHINCS+ is a hash-based signature scheme.} 

{A majority of these signature schemes are lattice-based. Therefore, a breakthrough result in the field of cryptanalysis of lattice-based cryptography could create a major dilemma in the transition from classical to post-quantum cryptography. }
Such incidents are not very rare. Some recent examples are Castryck~\textit{et al.}'s~\cite{An_Efficient_Key_Recovery_Attack_on_SIDH} attack on the post-quantum key-exchange mechanism based on supersingular isogeny Diffe-Hellman~\cite{Towards_quantum_resistant_cryptosystems_from_supersingular_elliptic_curve_isogenies} problem or Beullens's attack~\cite{beullens_rainbow} on post-quantum digital signature scheme Rainbow~\cite{rainbow}. Both of these schemes {were finalists} of the NIST's post-quantum standardization procedure. 
Therefore, diversification in the underlying hard problems ensures that if one of the cryptographic schemes is compromised, others may remain secure. Another problem of the currently standardized signature schemes is their very large signature sizes compared to classical signatures. This renders them almost impractical for real-world use cases like SSL/TLS {certificate chains}. Recognizing the critical importance of diversification and the practical use of digital signatures, NIST has recently issued an additional call~\cite{nist_additional_call} for post-quantum secure digital signatures. In this call, NIST emphasizes the importance {of small signature and fast verification to enhance practicality.} 

Linear Equivalence Signature Scheme (LESS)~\cite{LESS_is_More,LESS_Specification_Doc} is a submitted digital signature scheme aimed at increasing diversification and smaller signature and public key sizes. There are other code-based submissions like WAVE~\cite{WAVE_Specification_Doc}, enhanced pqsigRM~\cite{Enhanced_pqsigRM:Code-Based_Digital_Signature_Scheme_with_Short_Signature_and_Fast_Verification_for_Post-Quantum_Cryptography}, and CROSS~\cite{CROSS_Specification_Doc}. These schemes are based on the Syndrome Decoding Problem (SDP) for linear codes. The hardness of SDP relies on different variants of Information Set Decoding (ISD) algorithms. On the other hand, LESS has avoided the SDP, and it is the first cryptographic scheme {based on} the Code Equivalence Problem (CEP). The CEP asks to determine if two linear codes are equivalent to each other. {In the Hamming metric}, the notion of equivalence is linked to the existence of a monomial transformation, often termed the Linear Equivalence Problem (LEP).
 

 Due to the choice of this hard problem, the designers could choose parameters that lead to smaller key sizes without compromising security. The designers have also proposed different compression techniques to reduce the key sizes.
 LESS offers a balanced trade-off between {the combined public key and signature size} and the efficiency of signing and verification routines. Table~\ref{tab:comparison_of_sizes} in Appendix~\ref{appendix:comparison} compares the key sizes and efficiency of LESS and other code-based digital signature schemes.

We want to note that LESS first introduced the novel problem CEP or LEP for cryptographic constructions. {It uses a 3-round interactive sigma protocol between a prover and a verifier.} 
{Other signature schemes like MEDS and CROSS are also based on similar zero-knowledge identification schemes.} Multiple rounds of the identification scheme are used here, which is converted into a signature scheme using the Fiat-Shamir transformation.
However, using multiple rounds increases the signature size.
Here, we have noticed that all three signature schemes, LESS, CROSS and MEDS~\cite{Take_your_MEDS:_Digital_Signatures_from_Matrix_Code_Equivalence}, use the same compression technique that helped the designers ease the long-enduring bottleneck of large signature sizes in code-based cryptography. However, the implementation of this common compression technique has potential vulnerabilities against fault attacks that we identified in this work. 
Our primary motivation in this work is to uncover potential vulnerabilities against a wide spectrum of fault attacks and propose suitable countermeasures for the schemes LESS and CROSS that use the protocols having the same compression technique. We are confident that this work will help to improve the LESS and CROSS signature schemes and be useful in the evaluation of NIST's standardization procedure. {Further, we strongly believe that this will also be beneficial to other cryptographic signature schemes, such as MEDS, as it uses a similar technique.} Below, we briefly summarize our contributions.

\noindent
\textbf{Fault analysis of LESS digital signature:} {We have explored several fault attack surfaces of the LESS signature scheme that could be exploited by an adversary. We found different attack surfaces in the signing algorithm of LESS, and attack strategies that can be utilized on those attack surfaces. We observed that the designers of LESS proposed a technique to compress the signature size.} They used a binary tree called \textit{Reference Tree} to fulfil this purpose. We show that the modification of the values in the tree during the signing algorithm leaks information about the secret key as part of the output signature. We further use this information to recover the full secret key.

\noindent
{\textbf{Versatility of our fault attack:} Our attack assumes a single fault injection model. We want to note that our focus was to develop the theoretical framework to recover the secret after the fault injection.
In this regard, our attack can be realized using many different faults. Therefore, it is very versatile \textit{i.e.} not skewed in favour of the attacker.} In particular, we discuss the applicability of our attack using different types of faults, such as instruction skip, stuck-at-zero, and bit-flip. These types of faults can be realized using different mechanisms such as voltage glitch~\cite{Secret_External_Encodings_Do_Not_Prevent_Transient_Fault_Analysis}, Rowhammer~\cite{A_practical_key-recovery_attack_on_LWE-based_key-encapsulation_mechanism_schemes_using_Rowhammer,A_new_approach_for_rowhammer_attacks}, clock glitch~\cite{DBLP:journals/tches/BruinderinkP18,Fault_Attacks_on_CCA-secure_Lattice_KEMs}, laser fault injection~\cite{cryptoeprint:2022301}, electromagnetic fault injection~\cite{DBLP:journals/iacr/GenetKPM18,DBLP:journals/iacr/KunduCSKMV23} etc. 

\noindent
\textbf{Strong mathematical analysis:} We give detailed mathematical analysis to recover the secret key after the fault injections. We consider an arbitrary location for the fault injection, which is known to the attacker. Then, discuss the methods to recover the secret key in different scenarios. To further improve the effectiveness and practicality of our attack, we also provide a very effective method to remove noise from the experiments \textit{i.e.} differentiating between effective and ineffective faults. This is a non-trivial problem in any fault injection attack. We mathematically derived the expected amount of secret information that can be recovered from a single effective fault. 

\noindent \textbf{Application to other zero-knowledge based signature schemes: } Other code-based signature schemes in the NIST additional call for signatures such as CROSS~\cite{CROSS_Specification_Doc} and MEDS~\cite{Take_your_MEDS:_Digital_Signatures_from_Matrix_Code_Equivalence}, use a similar zero-knowledge framework as LESS. In these frameworks, the challenger and prover must communicate a series of challenges and responses for the soundness of schemes. This increases the signature size of the digital signature schemes designed using this framework. All these three signature schemes use a binary tree-based compression technique to reduce the signature size. As our attack targets this method, our attack strategy can also be extended to these schemes. We have explained this strategy for the CROSS signature scheme in this work.

\noindent
\textbf{Attack simulation:} We have an end-to-end fault attack simulation on the reference implementation of LESS and CROSS signature schemes. For LESS, we have simulated the attack in a way so that it can count the number of secret {matrix} recovered with one faulted signature, the number of faulted {signatures} required to recover the whole secret. {Also, our simulation induces fault with varying successful fault probability.} In both {schemes}, we modify a particular node of the binary tree structure and then recover the secret from the faulted signatures. {We have shown that if we inject fault in a specific location, then the entire secret can be recovered from a single effective fault signature for all the parameter sets of LESS except the parameter LESS-1s. For the CROSS signature scheme, only one effective faulted signature is enough to recover the complete secret for all parameters.}

\noindent
\textbf{Countermeasures:} Finally, we discuss different countermeasures that can prevent such attacks. We show that these countermeasures are effective against the single-fault attack models. Our first countermeasure removes the primary source of vulnerability \textit{i.e.} the generation of the \textit{Reference Tree}. This rather simple method increases the signature size. The second countermeasure modifies the \textit{Reference Tree} generation procedure such that the attack surfaces are eliminated. This method ensures that the signature sizes stay the same as the original LESS proposal~\cite{LESS_Specification_Doc}. {Lastly, we implemented the second countermeasure for LESS and compared its performance with the original LESS implementation. The performance cost of our second countermeasure is the same as the cost of the original LESS implementation.}

\section{{Preliminaries}}\label{sec:Preliminary}
$\ring{Z}{q}$ denotes the ring of integers modulo $q$. Additionally, $\field{F}{q}$ and $\group{F}{q}$ have been used to signify the field with $q$ elements and the multiplicative group {of this field $\field{F}{q}$}, respectively. The sets $\fieldvec{F}{q}{k}$ and $\fieldmat{F}{q}{k}{n}$ represent the collection of all vectors of size $k$ and all matrices of dimension $k\times n$ over the field $\field{F}{q}$, respectively. We use calligraphic uppercase ($\code{C}$) to denote a linear code. 

The lowercase letters ($a$) and uppercase letters ($A$) denote the scalars and the ordered set of scalars, respectively. $A^{c}$ represents the complement of the set $A$.  We use bold lowercase ($\vect{a}$) to denote vectors in any domain, and the $i$-th entry of the vector $\vect{a}$ is denoted by $\vecentry{a}{i}$. We denote the $i$-th standard basis 
as $\basisvec{i}$. The transpose of a vector $\vect{a}$ is denoted by $\vect{a}^T$.

The bold uppercase letters ($\mat{A}$) represent matrices. Let $\mat{A}$ be a matrix, then $\matentry{A}{i}{j}$ represents the $i,\ j$-th entry of the matrix $\mat{A}$. Also, $\matcol{A}{j}$ and $\matrow{A}{i}$ represent the  $j$-th column and $i$-th row of the matrix $\mat{A}$ respectively. Let $J\subset \mathbb{Z}_n$ be an ordered set of column indices of the matrix $\mat{A}$, then the notation $\matcol{A}{J}$ represents the submatrix of $\mat{A}$ formed by selecting columns with indices specified in the set $J$. Similarly, if $J$ is an ordered set of row indices of matrix $\mat{A}$, then the notation $\matrow{A}{J}$ represents the submatrix of $\mat{A}$ formed by selecting rows with indices specified in the set $J$. The transpose of a matrix $\mat{A}$ is denoted by $\mat{A}^T$. The inner product of two vectors $\vect{a}$ and $\vect{b}$ of same size is denoted by $\langle\vect{a},\ \vect{b}\rangle$ and is defined by $\sum_{i}{}\vecentry{a}{i}\vecentry{b}{i}$. The set of all invertible matrices of order $k$ over $\field{F}{q}$ is denoted by $GL_{k}(q)$.
\subsection{Definitions}
\begin{definition}[Monomial matrix]
An $n\times n$ matrix $\mat{A}$ is called a \textit{monomial matrix} if we can write $\mat{A}:=(\vecentry{u}{0}\basisvec{\pi(0)}~|~\vecentry{u}{1}\basisvec{\pi(1)}~|~\cdots~|~\vecentry{u}{n-1}\basisvec{\pi(n-1)})$. Here, $\vect{u}\in\mathbb{F}_q^n$, $\pi:\ring{Z}{n}\rightarrow \ring{Z}{n}$ is a permutation and $\vecentry{u}{j}\basisvec{\pi(j)}$ is $j$-th column of $\mat{A}$. We represent the monomial matrix $\mat{A}$ with the pair $(\pi,\ \vect{u})$. 
\end{definition}

\begin{definition}[Partial monomial matrix]
    An $n\times k$ matrix $\mat{B}$ is called a partial monomial matrix if we can write the matrix $\mat{B}:=(\vecentry{v}{0}\basisvec{\pi_{*}(0)}~|~\vecentry{v}{1}\basisvec{\pi_{*}(1)}~|~\cdots~|~\vecentry{v}{k-1}\basisvec{\pi_{*}(k-1)})$. Here, 
$n>k$, $\vect{v}\in\mathbb{F}_q^k$ and $\pi_{*}:\ring{Z}{k}\rightarrow\ring{Z}{n}$ is an injective mapping. We represent the \textit{partial monomial matrix} $\mat{B}$ with the pair $(\pi_{*},\ \vect{v})$.
\end{definition}  
We denote the set of all invertible monomial matrices of order $n$ and the set of all partial monomial matrices of order $n\times k$ over $\field{F}{q}$ by $M_{n}(q)$ and $M_{n,k}'(q)$ respectively.\\

\begin{definition}[Reduced Row-Echelon form] A matrix $\mat{A}$ of order $m\times n$ is said to be in Reduced Row-Echelon form (RREF) if the following conditions hold
\begin{enumerate}[label=\roman*., leftmargin=*]
    \item For each $0\leq i\leq m-1$, $0\leq j\leq n-1$, if the $i$-{th} row contains the first non-zero element at $j$-{th} position, then the first non-zero element of $(i+1)$-th row should be after the $j$-{th} position.
    \item The first non-zero element of any non-zero row is $1$.
    \item The leading element is the only non-zero element of that column.
\end{enumerate}
\end{definition}
We can transfer any matrix $\mat{A}$ to its \texttt{RREF} form by applying some elementary row operations~\cite{meyer2000matrix} on the matrix $\mat{A}$, and we denote this transformation by $\texttt{RREF}(\mat{A})$. Also, note that a matrix has a unique \texttt{RREF}. The first non-zero elements of $\texttt{RREF}(\mat{A})$ in each row are called \textit{pivots} and the columns that contain pivot are called \textit{pivot column} of the matrix $\texttt{RREF}(\mat{A})$. The remaining columns are called \textit{non-pivot columns}.

\begin{definition}[Lexicographically sorted order]
   Let $\vect{a}$ and $\vect{b}$ be two vectors of the same size over the field $\field{F}{q}$. We call the vectors $\vect{a}$ and $\vect{b}$ are in lexicographical order if $\vecentry{a}{i}<\vecentry{b}{i}$ holds, where $i$ is the first position where two vectors differ. We denote it as $\vect{a}<\vect{b}$. 
   Let there be $r$ vectors $\vect{v}_{0},\ \vect{v}_{1},\ \cdots,\ \vect{v}_{r-1}$ over the field $\field{F}{q}$. We call these vectors in lexicographically sorted order if, for any $0\leq i,\ j<r$, $\vect{v}_{i}<\vect{v}_{j}$ holds whenever $i<j$. 
\end{definition}
{A matrix G is lexicographically sorted if its columns are in ascending lexicographical order. }In this paper, the function $\texttt{LexMinCol}$ makes each column of input matrix $\mat{G}$ to lexicography sorted order by multiplying the inverse of the first non-zero element of that column and $\texttt{LexSort}$ function is used to sort the columns of $\mat{G}$ in lexicographically sorted order.

\begin{definition}[{Linear code}]
An $[n,\ k]$-linear code $\code{C}$ of length $n$ and dimension $k$ is a linear subspace of the vector space $\fieldvec{F}{q}{n}$. It can be represented by a matrix $\mat{G}\in\fieldmat{F}{q}{k}{n}$, which is called a generator matrix. For any $\vect{u}\in\fieldvec{F}{q}{k}$, the generator matrix $\mat{G}$ maps it to a code-word $\vect{u}\mat{G}\in\fieldvec{F}{q}{n}$.
\end{definition}  
\begin{definition}[{Linear code equivalence}]
    Let $\code{C}$ and $\code{C}'$ be two linear codes of length $n$ and dimension $k$ with generator matrices $\mat{G}$ and $\mat{G}'$ respectively. We call the codes $\code{C}$ and $\code{C}'$ linearly equivalent, if there exist matrices $\mat{Q}\in M_n (q)$, $\mat{S}\in GL_k(q)$ such that $\mat{G}' = \mat{SGQ}$.
\end{definition}
\begin{definition}[{Information Set (IS) of a Linear Code{~\cite{A_New_Formulation_of_the_Linear_Equivalence_Problem}}}]
{Let $\code{C}$ be a linear code with length $n$, and $J\subset \ring{Z}{n}$ be a set with cardinality $k$. Consider $\mat{G}$ as the generator matrix of code $\code{C}$. Define $J$ as an information set corresponding to $\mat{G}$ if inverse of $\mat{G}[*,\ J]$ exists $i.e.,$ $\mat{G}[*,\ J]$ is non-singular.}
\end{definition}

\subsection{LESS signature scheme}\label{sec:LESS}
The signature scheme LESS is based on the hardness of the Linear-code Equivalence Problem (LEP). LESS signature~\cite{LESS_Specification_Doc} uses a 3-round interactive sigma protocol~\cite{Identification_Protocols} between a prover and a verifier to establish the message's authenticity and the Fiat Shamir transformation~\cite{Identification_signatures_FiatShamir_transform} to transform this interactive protocol into a signature scheme. In this section, we describe the key generation and the signature algorithm of the digital signature LESS as it is most relevant to our work. Meanwhile, the verification algorithm is described in Appendix~\ref{sec:LESS_verification}. The description of the LESS signature involves some additional functions that we describe below.
 \begin{itemize}[leftmargin=*]
    \item $\texttt{CSPRNG}(\textit{seed},~\cdot~)$: This is a pseudo-random number generator, which takes a seed as input and outputs a pseudo-random string. The resulting output can be formatted according to preference, either as a string of seed values or a matrix. The uses of the function as $\texttt{CSPRNG}(\textit{seed},\ \mathbb{S}_{\tt RREF})$, $\texttt{CSPRNG}(\textit{seed},\ \mathbb{S}_{t,w})$ and $\texttt{CSPRNG}(\textit{seed},\ {M}_{n}(q))$ represents sampling a generator matrix in RREF, {sampling the fixed weight digest vector} and sampling a monomial matrix, respectively using the provided \textit{seed}.
    
    
    

    \item $\texttt{SeedTree}(\textit{seed},\ \textit{salt})$: This function generates a tree of height $\lceil\log\ t\rceil$. It begins with $\lambda$ bit input \textit{seed} and uses the \texttt{CSPRNG} function to generate $2\lambda$ bits. This long string is divided into two parts: the first $\lambda$ bits are used for the left child and the last $\lambda$ bits for the right child. The bits corresponding to each child are again fed into the \texttt{CSPRNG} with \textit{salt} to generate the next layer of the nodes in the tree. This process is repeated until the tree with height $\lceil\log\ t\rceil$ is constructed. 

    \item $\texttt{PrepareDigestInput}(\mat{G}, \ \mat{Q}')$: This function takes the matrices $\mat{G}$ which is in RREF and a monomial matrix $\mat{Q}'$ as inputs. Then computes $\mat{G}'$ as $(\mat{G}',\ pivot\_column)=\texttt{RREF}(\mat{GQ}'^{T})$. 
    Let $J=\left\{\alpha_{0},\ \alpha_{1},\ \cdots,\ \alpha_{k-1}\right\}$ be the set of pivot column indices{, which is essentially the information set (IS) of $\mat{G}'$}. Then, compute the partial monomial matrix $\overline{\mat{Q}}'$ and the matrix $\overline{\mat{V}}'$ as 
    $\overline{\mat{Q}}'=\mat{Q}'^{T}[*,\ J] \text{ and }\overline{\mat{V}}'= \texttt{LexSort}(\texttt{LexMinCol}( \mat{G}'[*,\ J^{c}]))\,.$
    After this computation, this function returns the partial monomial matrix $\overline{\mat{Q}}'$ and the matrix $\overline{\mat{V}}'$ as outputs.

    \item $\texttt{SeedTreePaths}(\vect{seed},\ \vect{f})$: Given a seed tree $\vect{seed}$ and a binary string $\vect{f}$ representing the leaves to be disclosed, this procedure derives which nodes of the seed tree should be disclosed so that the verifier can rebuild all the leaves which have been marked by the binary string. A detailed description of this function is given in Alg.~\ref{alg:SeedTreePaths}.


    \item \texttt{CompressRREF} and \texttt{CompressMono}:  \texttt{CompressRREF} function is used to compress a matrix $\mat{G}$ in RREF, and similarly \texttt{CompressMono} is used to compress a monomial matrix. Each compression procedure have corresponding expansion procedure that converts the compressed information to its proper matrix form. {Therefore, we can assume using or not using these function does not affect the functionality of key generation, signing or verification of LESS.}
    
\end{itemize}
\begin{algorithm}[!ht]
\caption{\texttt{LESS\_KeyGen($\lambda$)}~\cite{LESS_is_More,LESS_Hardware}}\label{alg:KeyGen}
\begin{algorithmic}[1]
\Require None
\Ensure $\text{SK}=(\textit{MSEED},\ \textit{gseed})$, $\text{PK}=(\textit{gseed},\ \mat{G}_{1},\ \dots,\ \mat{G}_{s-1})$

\State $\textit{MSEED}\xleftarrow{\$}\{0,\ 1\}^{\lambda}$
\State $\vect{mseed}\xleftarrow{}\texttt{CSPRNG}(\textit{MSEED}) \in \{0,\ 1\}^{(s-1)\lambda}$
\State $\textit{gseed}\xleftarrow{\$}\{0,1\}^{\lambda}$
\State $\mat{G}_{0}\leftarrow \texttt{CSPRNG}(\textit{gseed},\ \mathbb{S}_{\texttt{RREF}} )$
\For{$i=1;\ i<s;\ i=i+1$}
    \State $\mat{Q}_{i}\xleftarrow{}\texttt{CSPRNG}(\vecentry{mseed}{i},\ M_{n}(q))$
    \State $(\mat{G}_{i},\ pivot\_column)\leftarrow \texttt{RREF}(\mat{G}_{0}(\mat{Q}_{i}^{-1})^{T})$
    \State $\text{PK}[i]\leftarrow \texttt{CompressRREF}(\mat{G}_{i},~pivot\_column)$
\EndFor
\State Return $(\text{SK},\ \text{PK})$
\end{algorithmic}
\end{algorithm}
\subsubsection{Key Generation of LESS: }
It is presented in Alg.~\ref{alg:KeyGen}. Given a security parameter $\lambda$, the two outputs of this algorithm are the secret key $\text{SK}$ and the public key $\text{PK}$. The first component of the secret key is the master key $\textit{MSEED}\in \left\{0,\ 1\right\}^{\lambda}$. Using the \texttt{CSPRNG} function, the vector {$\vect{mseed}\in\left\{0,\ 1\right\}^{(s-1)\lambda}$} is generated from the $\textit{MSEED}$, which contains $s-1$ many $\lambda$-bit binary strings. Now, the $i$-th secret monomial matrix $\mat{Q}_{i}$ is generated from {$\vecentry{mseed}{i}\in \left\{0,\ 1\right\}^{\lambda}$}. Note that these generated $\mat{Q}_i$'s are all secret monomial matrices. Also, the seed $\textit{gseed}$ is employed in the generation of the public matrix $\mat{G}_{0}.$ The remaining part of the public key consists of the matrices $\mat{G}_{i}$ for $1\leq i \leq s-1$, which are generated using the process described in Alg.~\ref{alg:KeyGen}\footnote{{For simplicity and compactness, we follow the implementation of LESS instead of the specification document}}.
\begin{algorithm}[]
\caption{\texttt{LESS\_Sign}(${m},\ \text{SK}$)}\label{alg:Signature}
\begin{algorithmic}[1]
\Require Message $\textit{m}\in\ring{Z}{2}^\textit{len}$ and secret key $\text{SK}=(\textit{MSEED},\ \textit{gseed})$.
\Ensure The signature $\tau=(\textit{salt},\ \textit{cmt},\ \vect{TreeNode},\ \vect{rsp})$.
\State {$\vect{mseed}\xleftarrow{}\texttt{CSPRNG}(\textit{MSEED})\in \left\{0,\ 1\right\}^{(s-1)\lambda}$}
\State $\textit{EMSEED}\xleftarrow{\$}\{0,\ 1\}^{\lambda}$, $\textit{salt}\xleftarrow{\$}\{0,\ 1\}^{\lambda}$
\State $ \vect{seed}\xleftarrow{}\texttt{SeedTree}(\textit{EMSEED},\ \textit{salt})$
\State $ \vect{ESEED}=\text{Leaf nodes of the }\vect{seed}$
\State $\mathbf{G}_{0}\leftarrow \texttt{CSPRNG}(\textit{gseed},\ \mathbb{S}_\texttt{RREF})$

\For{$i=0;\ i<t;\ i=i+1$}
    \State $\widetilde{\mat{Q}}_{i}\xleftarrow{}\texttt{CSPRNG}(\vecentry{ESEED}{i},\ M_{n}(q))$
    \State $(\overline{\mat{Q}}_{i},\ \overline{\mat{V}}_{i})\xleftarrow{} \texttt{PrepareDigestInput}(\mat{G}_{0},\ \widetilde{\mat{Q}}_{i})$ 
\EndFor
\State $\textit{cmt} \leftarrow \texttt{H}(\overline{\mat{V}}_{0},\ \dots,\ \overline{\mat{V}}_{t-1},\ \textit{m},\ \textit{len},\ \textit{salt})$
\State $\vect{d}\leftarrow \texttt{CSPRNG}(\textit{cmt},\ \mathbb{S}_{t,w})$ 

\For{$i=0;\ i<t;\ i=i+1$}
\If{$\vecentry{d}{i}=0$} \State $\vecentry{f}{i}=0$
\Else \State $\vecentry{f}{i}=1$ 
\EndIf
\EndFor

\State $\vect{TreeNode}\leftarrow \texttt{SeedTreePaths}(\vect{seed},\ \vect{f})$\Comment{(Alg.~\ref{alg:SeedTreePaths})}

\State $k=0$
\For{$i=0;\ i<t;\ i=i+1$}
    \If{$\vecentry{d}{i}\neq 0$}
    \State $j=\vecentry{d}{i}$
    \State $\mat{Q}_{j}\xleftarrow{}\texttt{CSPRNG}(\vecentry{mseed}{j},\ {M}_{n}(q))$
    \State $\mat{Q}^*_{k}\leftarrow \mat{Q}_{j}^{T}\overline{\mat{Q}}_{i}$
    \State $\vecentry{rsp}{k}\leftarrow \texttt{CompressMono}(\mat{Q}^*_{k})$
    \State $k=k+1$
    \EndIf
\EndFor
\State Return $\tau=(\textit{salt},\ \textit{cmt},\ \vect{TreeNode},\ \vect{rsp})$
\end{algorithmic}
\end{algorithm}
\subsubsection{Signature algorithm of LESS:}\label{sec:LESS_sign}
The signature algorithm shown in Alg.~\ref{alg:Signature} takes a message string $\textit{m}$ of length $\textit{len}$ and the secret key $\text{SK}=(\textit{MSEED},\ \textit{gseed})$ as inputs and returns a corresponding signature $\tau$. 
The main secret key component of $\text{SK}$ is the master seed $MSEED$. All the $s-1$ monomial matrices $\mat{Q}_{j}$ are generated from the $MSEED$ and used to produce signatures. That is, instead of having information of $MSEED$, if we have the information of all of $s-1$ monomial matrices $\mat{Q}_{j}$, then we can construct the same valid signature. Therefore, these monomial matrices $\mat{Q}_{j}$ are considered equivalent to the secret key component $MSEED$. To reduce the signature size, the authors of LESS have incorporated a method involving tree construction. We explain this process briefly here.
\begin{figure}[!ht]
    \centering
    \includegraphics[width=.4\linewidth]{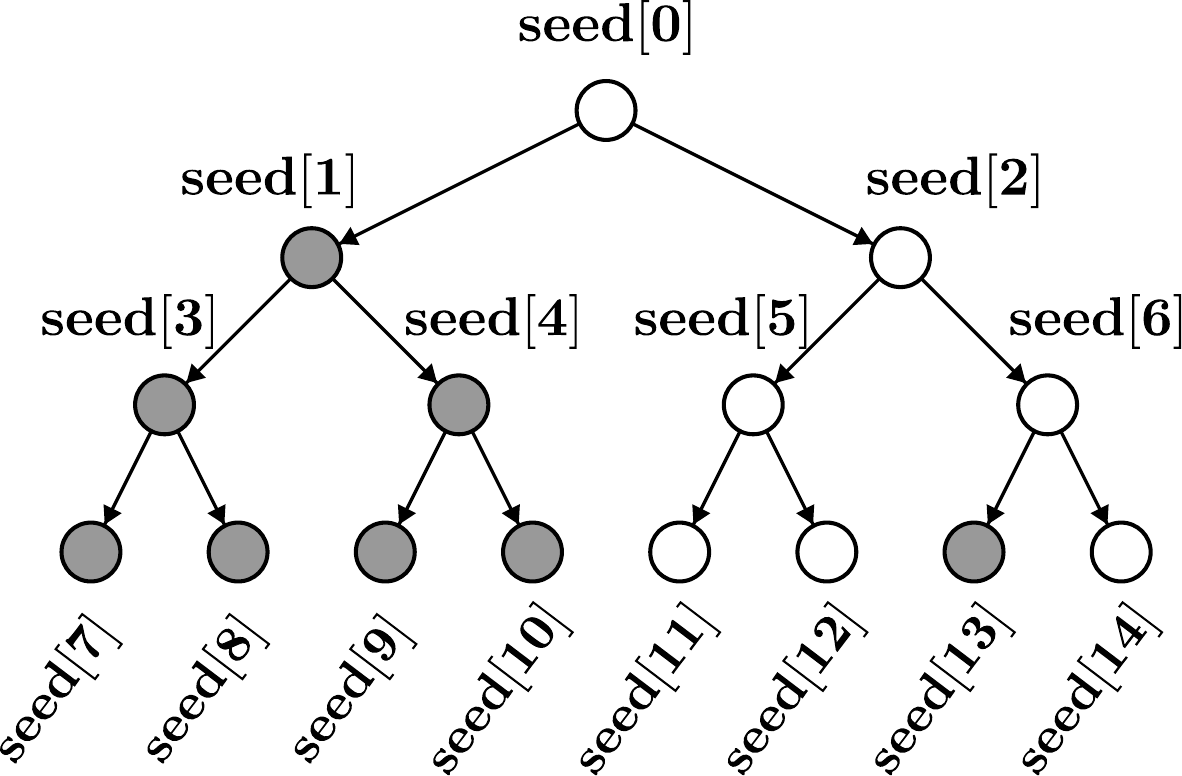}
    \caption{Example of seed tree}
    \label{fig:seedtree}
\end{figure}

First, we outline the procedure for generating a set of $t$ ephemeral monomial matrices represented by $\widetilde{\mat{Q}}_{0},\ \widetilde{\mat{Q}}_{1},\ \cdots,\ \widetilde{\mat{Q}}_{t-1}$ through the generation of $t$ random ephemeral seeds denoted as $\vecentry{ESEED}{i}$ for $0\leq i< t$. The process involves the following steps:
\begin{itemize}[leftmargin=*]
    \item Start by sampling a random master seed $\textit{EMSEED}\xleftarrow{\$}\left\{0,\ 1\right\}^{\lambda}$.
    \item Build a tree of seed nodes using \texttt{SeedTree} procedure, with output tree named $\vect{seed}$. The height and the number of leaf nodes of the output tree are $\lceil \log(t)\rceil$ and $2l=2^{\lceil \log(t)\rceil}$ respectively where the input seed is the master seed $\textit{EMSEED}$.
    \item  Select the first $t$ leaf nodes of the $\vect{seed}$ as the ephemeral seeds $\vecentry{ESEED}{i}$, where  $\vecentry{ESEED}{i}=\vecentry{seed}{2l-1+i}$ for $0\leq i<t$.
    \item Using \texttt{CSPRNG} function $\widetilde{\mat{Q}}_i$ is prepared for each $\vecentry{ESEED}{i}$.
\end{itemize} 
Lines 6-8 of Alg.~\ref{alg:Signature}, correspond to generating the partial monomial matrices $\overline{\mat{Q}}_i$, the matrices $\overline{\mat{V}}_i$ having the information of the non-pivot part corresponding to the matrix $\mat{G}_{0}\widetilde{\mat{Q}}_{i}^{T}$. Using the information of all $\overline{\mat{V}}_{i}$ matrices, message $m$, message length $len$ and $salt$, the digest $\vect{d}\in \ring{Z}{s}^{t}$ is prepared. This digest vector $\vect{d}$ has fixed weight $w$, where weight of the vector $\vect{d}$ is defined as $wt(\vect{d}):=|\left\{i:~\vecentry{d}{i}\neq 0\right\}|$.
We will briefly discuss the \texttt{SeedTreePaths} procedure in Alg.~\ref{alg:SeedTreePaths}, as our attack is based on exploting this procedure. This \texttt{SeedTreePaths}~(Alg.~\ref{alg:SeedTreePaths}) helps to reduce the size of the signature. Finally, the signature will return $\mat{Q}_{\vecentry{d}{i}}^T\overline{\mat{Q}_i}$ whenever $\vecentry{d}{i}\neq 0$, and it also reveals the seed nodes so that $\widetilde{\mat{Q}}_i$ can be generated from the revealed seed nodes for all $i$ such that $\vecentry{d}{i}=0$. Note that whenever we try to return $\widetilde{\mat{Q}}_i$, it is enough to return $\vecentry{ESEED}{i}$'s instead. Also, having the information of any ancestor node of the seed $\vecentry{ESEED}{i}(=\vecentry{seed}{2l-1+i})$, we can get the information of $\vecentry{ESEED}{i}$. This is the idea behind the minimization of the number of seeds that are to be sent. This minimized set is returned as $\vect{TreeNode}$. Consider the example in Fig.~\ref{fig:seedtree}, where the leaf nodes are $\vecentry{ESEED}{i}$'s and the shaded leaf nodes represent all those positions where $\vect{d}$ takes the value 0 $i.e.$, these are the $\vecentry{ESEED}{i}$'s that are to be revealed. Observe that revealing only $\vect{TreeNode}=(\vecentry{seed}{1},\ \vecentry{seed}{13})$ is enough, as the required $\vecentry{ESEED}{i}$'s can be regenerated at the time of verification. Consequently, this minimizes the signature size.

Now, in lines 18-24 of Alg.~\ref{alg:Signature}, the $\vect{rsp}$ is prepared by appending the partial monomial matrices $\mat{Q}_{\vecentry{d}{i}}^T\overline{\mat{Q}}_i$ for all $i$ such that $\vecentry{d}{i}$ is non-zero. Since the length and the weight of the fixed weight digest $\vect{d}$ are $t$ and $w$ respectively, the vector $\vect{d}$ has exactly $w$ many non-zero elements and $t-w$ many zero elements. Therefore, the signature will contain the component $\vect{rsp}$ having exactly $w$ many matrices of the form $\mat{Q}_{\vecentry{d}{i}}^T\overline{\mat{Q}}_i$. After all of these computations, $(salt,$ $cmt,$ $\vect{TreeNode},$ $\vect{rsp})$ is generated as the signature.
\subsection{Parameter set}
There are {three security levels} of LESS~\cite{LESS_Specification_Doc} and their corresponding parameter sets, which are shown in Table~\ref{tab:parameters}.
Here, the code parameters are given by 
$n$: the length of the code, $k$: the dimension of the code, $q$: prime modulus corresponding to the finite field $\mathbb{F}_q$, $2l$: the number of leaf nodes of the seed tree, where $2l=2^{\lceil\log t\rceil}$, $t$: the length of the digest $\vect{d}$, $w$: the fixed weight of the digest $\vect{d}$ and $s$: $s-1$ is the number of secret monomial matrices. 
{According to the LESS documentation {\cite{LESS_Specification_Doc}}, multiple parameter sets are defined for each security level of LESS, and the optimization criteria for each of these parameter sets are different. The "b" version (e.g., LESS-1b) refers to the parameter set with balanced public key and signature size, the "s" (e.g., LESS-1s) version refers to the parameter set with smaller signature size, and the "i" (only LESS-1i) version refers to the parameter set with intermediate public key and signature size.}

\begin{table}[!ht]
\centering
\begin{tabular}{c|c|ccccccc|c|c}
\hline
\multirow{2}{*}{\begin{tabular}[c]{@{}c@{}}Security\\  level\end{tabular}} & \multirow{2}{*}{\begin{tabular}[c]{@{}c@{}}Parameter\\  set\end{tabular}} & \multicolumn{7}{c|}{Parameters} & \multirow{2}{*}{\begin{tabular}[c]{@{}c@{}}{Public key (PK)}\\  {(KiB)}\end{tabular}} & \multirow{2}{*}{\begin{tabular}[c]{@{}c@{}}{Signature ($\tau$)}\\ {(KiB)}\end{tabular}} \\ \cline{3-9}
 &  & $n$ & $k$ & $q$ & $l$ & $t$ & $w$ & $s$ &  &  \\ \hline
\multirow{3}{*}{1} & LESS-1b & \multirow{3}{*}{252} & \multirow{3}{*}{126} & \multirow{3}{*}{127} & \multirow{3}{*}{128} & 247 & 30 & 2 & {13.7} & {8.1} \\
 & LESS-1i &  &  &  &  & 244 & 20 & 4 & {41.1} & {6.1} \\
 & LESS-1s &  &  &  &  & 198 & 17 & 8 & {95.9} & {5.2} \\ \hline
\multirow{2}{*}{3} & LESS-3b & \multirow{2}{*}{400} & \multirow{2}{*}{200} & \multirow{2}{*}{127} & \multirow{2}{*}{512} & 759 & 33 & 2 & {34.5} & {18.4} \\
 & LESS-3s &  &  &  &  & 895 & 26 & 3 & {68.9} & {14.1} \\ \hline
\multirow{2}{*}{5} & LESS-5b & \multirow{2}{*}{548} & \multirow{2}{*}{274} & \multirow{2}{*}{127} & 1024 & 1352 & 40 & 2 & {64.6} & {32.5} \\
 & LESS-5s &  &  &  & 512 & 907 & 37 & 3 & {129.0} & {26.1} \\ \hline
\end{tabular}
\caption{Parameter set of LESS~\cite{LESS_Specification_Doc} for different security levels}
\label{tab:parameters}
\end{table}

%
\section{Our Work: Fault analysis of LESS}\label{sec:OurWork} 
 {One of the strongest physical attacks on the digital signature schemes is to recover the secret or signing key, as the adversary can compute any valid message and signature pair using the recovered signing key.} 
 {In general, only the key generation and the signing algorithm involve the secret key. However, only the signing algorithm uses the long-term secret key (the same secret key is used multiple times), making it most suitable for performing a physical attack}~\cite{DBLP:journals/tches/BruinderinkP18,cryptoeprint:2017/1014,fault_rsa,flush_gauss_bliss}.
In this work, our objective is to mount a fault attack on the zero-knowledge based digital signature schemes. {In this attack model, the adversary would query the faulted signature oracle (which outputs a signature with some injected faults) multiple times.}
In this section, we will progressively describe our fault attack strategy to recover the secret monomial matrices for the LESS signature scheme. Later in Section~\ref{sec:CROSS}, we show that the same attack strategy can be employed in other zero-knowledge based signature schemes, such as CROSS, to recover the signing key.

\subsection{An observation on LESS}\label{subsec:Observation}
LESS signature algorithm presented in Alg.~\ref{alg:Signature} returns either the information of the monomial matrix $\widetilde{\mat{Q}}_j$ or the multiplication  $\mat{Q}_{\vecentry{d}{j}}^{T}\overline{\mat{Q}}_j$ for any $j\in\ring{Z}{t}$. Here, $\overline{\mat{Q}}_j$ is a \textit{partial monomial} matrix that is generated from the matrix $\widetilde{\mat{Q}}_{j}$ by using
the \texttt{PrepareDigestInput} function. If we manage to get a pair $(\widetilde{\mat{Q}}_{j},\ \mat{Q}_{\vecentry{d}{j}}^{T}\overline{\mat{Q}}_{j})$ for some $\vecentry{d}{j}\neq 0$, then we can construct the pair $(\overline{\mat{Q}}_{j},\ \mat{Q}_{\vecentry{d}{j}}^{T}\overline{\mat{Q}}_{j})$. This pair $(\overline{\mat{Q}}_{j},\ \mat{Q}_{\vecentry{d}{j}}^{T}\overline{\mat{Q}}_{j})$ leaks some information of matrix $\mat{Q}_{\vecentry{d}{j}}^{T}$ that is directly follows from the following lemma.

 \begin{lemma}\label{lemma:partial_monomial}
 {
 Let $\mat{A}=(\pi,\ \vect{u})\in M_{n}(q)$ be a monomial matrix and $\mat{B}=(\pi',\ \vect{u'})\in M_{n, k}'(q)$ be a partial monomial matrix. Let $\mat{C}=(\pi'',\ \vect{u''})\in M_{n, k}'(q)$ be the partial monomial matrix defined by $\mat{C}=\mat{A}^T\mat{B}$. Given the matrices $\mat{B}$ and $\mat{C}$, we can compute exactly $k$ many columns of the monomial matrix $\mat{A}^T$. More specifically, for all $0\leq j<k$, we can compute $\pi^{-1}(\pi'(j))$ and $\vect{u}[\pi^{-1}(\pi'(j))]$.    
 }
 \end{lemma}
 \begin{proof}
 {
 For the monomial matrix $\mat{A}$ represented by $(\pi,\ \vect{u})$, the transpose of $\mat{A}$ is the following matrix
 }
 \begin{align*}
     \mat{A}^T=[\vecentry{u}{\pi^{-1}(0)}\basisvec{\pi^{-1}(0)}~|~\vecentry{u}{\pi^{-1}(1)}\basisvec{\pi^{-1}(1)}~|~\cdots~|~\vecentry{u}{\pi^{-1}(n-1)}\basisvec{\pi^{-1}(n-1)}]
 \end{align*}
 {
 The multiplication of the monomial matrix $\mat{A}^T$ with the partial monomial matrix $\mat{B}$ is given by 
 }
 \begin{align*}
     \mat{A}^T\mat{B}=[\vecentry{u}{\pi^{-1}(\pi'(0))}\vecentry{u'}{0}\basisvec{\pi^{-1}(\pi'(0))}~|~\cdots~|~\vecentry{u}{\pi^{-1}(\pi'(k-1))}\vecentry{u'}{k-1}\basisvec{\pi^{-1}(\pi'(k-1))}]\
 \end{align*}
 {
 Since $\mat{C}=\mat{A}^T\mat{B}$, so for all $0\leq j<k$ we have $\matcol{C}{j}=(\mat{A}^T\mat{B})[*,\ j]$, which implies $\vecentry{u''}{j}\basisvec{\pi''(j)}=\vecentry{u}{\pi^{-1}(\pi'_{*}(j)}\vecentry{u'}{j}\basisvec{\pi^{-1}(\pi'_{*}(j))}$. This gives us the following 
 }
 \begin{align*}
 \begin{aligned}
     &\vecentry{u''}{j}= \vecentry{u}{\pi^{-1}(\pi'(j))}\vecentry{u'}{j}\\
     &\pi''(j)=\pi^{-1}(\pi'(j)) 
\end{aligned}
\end{align*}
{
 Since $\mat{B}$ and $\mat{C}$ are known, we have the information of each $\pi'(j)$, $\pi''(j)$, $\vecentry{u'}{j}$ and $\vecentry{u''}{j}$ where $0\leq j<k$. Therefore for all $0\leq j<k$ we have,
 }
\begin{align}
\begin{aligned}
     &\vecentry{u}{\pi^{-1}(\pi'(j))}=\vecentry{u''}{j}(\vecentry{u'}{j})^{-1}\\
     &\pi^{-1}(\pi'(j))=\pi''(j) \label{equ:partialMono}
\end{aligned}
\end{align}
{
Note that we have computed $\pi'(j)$-th column of the matrix $\mat{A}^T$ for all $0\leq j<k$.
}
\qed
 \end{proof}

{
For simplicity, in this part, we will use consider the matrices $\widetilde{\mat{Q}}_j$, $\overline{\mat{Q}}_j$ and $\mat{Q}_{\vecentry{d}{j}}$ as the matrices $\widetilde{\mat{Q}}$, $\overline{\mat{Q}}$ and $\mat{Q}$ respectively. 
Recall the \texttt{prepareDigestInput} function, it was taking $\mat{G}_0$ and a monomial matrix $\widetilde{\bm{Q}} = (\widetilde{\pi},\ \widetilde{\vect{\upsilon}})$ as input and $\overline{\mat{Q}}$ is one of the outputs of the function. The $\overline{\mat{Q}}$ is computed in a way that $\mat{G}_0\overline{\mat{Q}}=\mat{G}_0(\widetilde{\mat{Q}})^T[*,J^\dagger]$, where $J^\dagger$ is an IS of $\mat{G}_0(\widetilde{\mat{Q}})^T$. From the definition of IS, we can say that $\mat{G}_0\overline{\mat{Q}}$ is a non-singular matrix. Observe that,}
\begin{align}
    \mat{G}_0\overline{\mat{Q}}=[\overline{\vect{\upsilon}}_{0}\cdot \vect{g}_{\overline{\pi}(0)}~|~
    \overline{\vect{\upsilon}}_{1}\cdot \vect{g}_{\overline{\pi}(1)}~|~
    \cdots~|~
    \overline{\vect{\upsilon}}_{k-1}\cdot \vect{g}_{\overline{\pi}(k-1)}
    ]\label{equ:IS}
\end{align}
{Where $\overline{\mat{Q}}$ is a partial monomial matrix represented by $(\overline{\pi},\ \overline{\vect{\upsilon}})$. Since the matrix representation in {Eq.~\ref{equ:IS}} is non-singular, the set $J=\{\overline{\pi}(i)~:~i\in\mathbb{Z}_k\}$ is the IS of $\mat{G}_0$. 

Now consider we are given the pair $(\widetilde{\mat{Q}},\ \mat{Q}^T\overline{\mat{Q}})$, where $\mat{Q}$ represented by $(\pi,\ \vect{\upsilon})$ and $\overline{\mat{Q}}$ is generated from $\widetilde{\mat{Q}}$ using the function \texttt{prepareDigestInput}. Now, $\mat{Q}^T\overline{\mat{Q}}$ is a partial monomial matrix and let it be represented by $(\pi_*,\ \vect{\upsilon}_*)$ then from Lemma~{\ref{lemma:partial_monomial}}, we can write that for any $j\in\mathbb{Z}_k$} 
\begin{align}
\begin{aligned}
    &\pi^{-1}(\overline{\pi}(i))=\pi_*(i)\\
    &\vect{\upsilon}[\pi^{-1}(\overline{\pi}(i))]=\vect{\upsilon}_*[i](\overline{\vect{\upsilon}}[i])^{-1}\cdot\label{eq:recover}
\end{aligned}
\end{align}
{This Eq.~{\ref{eq:recover}} gives us the partially recovered secret $i.e.$ only $k$ many columns of $\mat{Q}^T$. According to the definition of $\overline{\pi}$, the set $\{\overline{\pi}(i)~:~i\in\mathbb{Z}_k\}$ is the set $J$ which is the information set of $\mat{G}_0$. Now, if $\mat{Q}$ is a secret monomial then from the key generation of LESS, we can say that $\widehat{\mat{G}} = \text{RREF}(\mat{G}_0(\mat{Q}^T)^{-1})$ is a part of the public key. We can further write $\widehat{\mat{G}} = \mat{S}\mat{G}_0(\mat{Q}^T)^{-1}$ for some non-singular matrix $\mat{S}$. Consider $\widehat{\mat{G}} = [\widehat{\vect{g}}_0~|~\widehat{\vect{g}}_1~|~\cdots~|~\widehat{\vect{g}}_{n-1}]$ then for all $i\in\mathbb{Z}_n$ we have $\widehat{\vect{g}}_i=\mat{S}\cdot \left((\vect{\upsilon}[{i}])^{-1}\cdot\vect{g}_{\pi(i)}\right)$ which implies that for all $i\in\mathbb{Z}_n$, }
\begin{align}
    \widehat{\vect{g}}_{\pi^{-1}(i)}=\mat{S}\cdot \left((\vect{\upsilon}[{\pi^{-1}(i)}])^{-1}\cdot\vect{g}_{i}\right) 
\label{equ:G_hat}
\end{align} 
{
Consider that the set $J$ have the elements $j_0,\ j_1,\ \cdots,\ j_{k-1}$, and we take the matrix $\mat{G}^* = [\widehat{\vect{g}}_{\pi^{-1}(j_0)}~|~\widehat{\vect{g}}_{\pi^{-1}(j_1)}~|~\cdots~|~\widehat{\vect{g}}_{\pi^{-1}(j_{k-1})}]$ and also take the matrix}
$$\mat{G}' = [
(\vect{\upsilon}[\pi^{-1}(j_0)])^{-1}\cdot\vect{g}_{j_0}~|~
(\vect{\upsilon}[\pi^{-1}(j_1)])^{-1}\cdot\vect{g}_{j_1}~|~
\cdots~|~
(\vect{\upsilon}[\pi^{-1}(j_{k-1})])^{-1}\cdot\vect{g}_{j_{k-1}}]$$
{From {Eq.~\ref{equ:G_hat}}, we have $\mat{G}^*=\mat{S}\mat{G}'$ and since $J$ is an IS of $\mat{G}_0$, so $\mat{G}'$ is a non-singular matrix. Also $\mat{G}'$ and $\mat{G}^*$ are both computable as for each $j\in J$, $\pi^{-1}(j)$ and $\vecentry{v}{\pi^{-1}(j)}$ are already recovered. Therefore, we can compute $\mat{S} = \mat{G}^*\cdot(\mat{G}')^{-1}$. Finally, we have $\mat{S}^{-1}\widehat{\mat{G}} = \mat{G}_0(\mat{Q}^T)^{-1}$, where $\mat{S}$, $\mat{G}_0$ and $\widehat{\mat{G}}$ are known. Using {Alg.~\ref{alg:getColPerm}}, we can recover the full secret.
}
\begin{algorithm}[!ht]
\caption{{\texttt{getColumnPermutation}($\widehat{\mat{G}}, \mat{G}_0, \mat{S}$)}}\label{alg:getColPerm}
\begin{algorithmic}[1]
\Require {The partially recovered secret $\pi:J_*\to J$ and $\vecentry{v}{j}$ $\forall j\in J_*$, where $J_* = \{\pi^{-1}(i)~:~i\in J\}$, public information $\mat{G}_0$ and $\widehat{\mat{G}}$, recovered matrix $\mat{S}$}
\Ensure {Outputs rest of the secret $\pi:J^c_*\to J^c$ and $\vecentry{v}{j}$ $\forall j\in J^c_*$}
\State {$[\vect{g}_0~|~\vect{g}_1~|~\cdots~|~\vect{g}_{n-1}]\gets \mat{G}_0$}
\State {$[\widehat{\vect{g}}_0~|~\widehat{\vect{g}}_1~|~\cdots~|~\widehat{\vect{g}}_{n-1}]\gets \widehat{\mat{G}}$}
\For{{$j\in J^c$}}
    \For{{$i\in J^c_*$}}
        \For{{$a\in \mathbb{F}_q$}}
            \If{{$\vect{g}_j = a\cdot (\mat{S}^{-1} \widehat{\vect{g}}_i)$}}
                \State {assign $\pi(i) \gets j$}
                \State {assign $\vect{v}[i]\gets a$}
            \EndIf
        \EndFor
    \EndFor
\EndFor
\end{algorithmic}
\end{algorithm}

{We can conclude that from one pair $(\widetilde{\mat{Q}}_{j},\ \mat{Q}_{\vecentry{d}{j}}^{T}\overline{\mat{Q}}_{j})$, we can recover the secret monomial matrix $\mat{Q}_{\vecentry{d}{j}}^{T}$, where $\vecentry{d}{j}\neq 0$.} However, we will not receive the pair $(\widetilde{\mat{Q}}_{j},\ \mat{Q}_{\vecentry{d}{j}}^{T}\overline{\mat{Q}}_{j})$ if the signatures are generated by executing the signing algorithm properly. Therefore, we must find strategies to disrupt the normal flow of execution to help us get such pairs. Also, note that, if the number of secret monomial matrices$(s-1)$ is greater than one, then receiving only one such pair is not enough to retrieve all the secret monomials. So, we may require multiple faulted signatures to receive several such pairs and finally recover all the secret monomial matrices. All of these analysis are briefly described in the later sections. 

\subsection{Identification of attack surfaces}\label{subsec:attackvector}
As we observed that having one pair of the form $(\widetilde{\mat{Q}}_{j},\ \mat{Q}_{\vecentry{d}{j}}^{T}\overline{\mat{Q}}_{j})$ is enough to recover the secret matrix $\mat{Q}_{\vecentry{d}{j}}^{T}$, where $\vecentry{d}{j}\neq 0$. 
Also, observe that, in LESS, there are $s-1$ secret monomial matrices $\mat{Q}_{i}$ for $1\leq i\leq s-1$, and $t$ ephemeral monomial matrices $\widetilde{\mat{Q}}_j$ for $0\leq j<t$ as described in Section~\ref{sec:LESS_sign}. Hence, our goal is to find at least one pair of the form $(\widetilde{\mat{Q}}_{j},\ \mat{Q}_{\vecentry{d}{j}}^{T}\overline{\mat{Q}}_{j})$, where $1\leq \vecentry{d}{j}\leq s-1$ and $0\leq j<t$ by manipulating the signing algorithm.

{Note that, LESS is a code-based signature scheme based on the sigma-protocol with Fiat-Shamir transformation. In Alg.~{\ref{alg:Signature}}, the signer generates the random challenge $\vect{d}$ (fixed weight digest), from commitment (\textit{cmt}) using the pseudo-random function \texttt{CSPRNG}. Any fault injection before the challenge generation may modify the challenge value, but that is an output of a pseudo-random function. This would not help, as we need to recover the secret key. Therefore, we have targeted to inject a fault after the generation of $\vect{d}$.}

\subsubsection{\textbf{Modification of the vector $\vect{d}$: }}
As we can see from Alg.~\ref{alg:Signature}, the digest $\vect{d}$ ($\vecentry{d}{i}$ for $0\leq i<t$) value decides whether $\widetilde{\mat{Q}}_{i}$ is revealed or $\mat{Q}_{\vecentry{d}{i}}^T\overline{\mat{Q}}_i$ is revealed. Therefore, the most obvious target for fault injection is the digest $\vect{d}$ to reveal both $\widetilde{\mat{Q}}_{i}$ and $\mat{Q}_{\vecentry{d}{i}}^T\overline{\mat{Q}}_i$ for some $i$.
If we modify some value $\vecentry{d}{i}$ of $\vect{d}$ (line 11 in Alg.~\ref{alg:Signature}) from $0$ to some non-zero value $r$ by injecting fault, then we will get the information of $\mat{Q}_{r}^{T}\overline{\mat{Q}}_{i}$ instead of getting information of $\widetilde{\mat{Q}}_{i}$. Similarly, if we change the value of $\vecentry{d}{i}$ from non-zero value $r$ to $0$, then we will get the information of $\widetilde{\mat{Q}}_{i}$ instead of getting information about $\mat{Q}_{r}^{T}\overline{\mat{Q}}_{i}$. 
In both cases, we do not receive $\widetilde{\mat{Q}}_{i}$ and $\mat{Q}_{r}^{T}\overline{\mat{Q}}_{i}$ together. Therefore, modifying the $\vect{d}$ value does not satisfy our purpose.

\begin{algorithm}[!ht]
\caption{\texttt{SeedTreePaths}}\label{alg:SeedTreePaths}
\begin{algorithmic}[1]
\Require The \textit{Seed Tree} $\vect{seed}$ and the vector $\vect{f}$.
\Ensure Outputs $\vect{TreeNode}$ a subset of \textit{Seed Tree} which consists only the $\mathtt{seed}_i$'s that does not correspond to $\vecentry{f}{i}=1$.
\For{$i=0;\ i<4l-1;\ i=i+1$}
\State $\vecentry{x}{i} =0$
\EndFor

\State $\vect{x}\leftarrow \texttt{compute\_seeds\_to\_publish}(\vect{f},\ \vect{x})$ \Comment{(Alg.~\ref{alg:ComputeSeedToPublish})}
\State $j = 0$
   \For{$i = 0;\ i<4l-1;\ i=i+1$} 
   \If{$(\vecentry{x}{i}=0\text{ and }\vecentry{x}{Parent(i)}=1)$}
   \State $\vecentry{TreeNode}{j} = \vecentry{seed}{i}$
 \State $j=j+1$
 \EndIf
\EndFor
\State return $\vect{TreeNode}$
\end{algorithmic}
\end{algorithm}
One might think of using the cases $\vecentry{d}{i} = 0$ bypassing the check $\vecentry{d}{i}\neq 0$ (line 19) using a fault. However, $\vecentry{mseed}{0}$ does not exist and might cause an error during execution. Therefore, modifying anything from lines 18-24 would not benefit us. Now, we analyse the remaining steps (lines 11-16) of Alg.~\ref{alg:Signature}. In these steps, we can modify the value of the vector $\vect{f}$. Also, the $\texttt{SeedTreePaths}$ algorithm is another potential candidate for fault injection, which is presented in Alg.~\ref{alg:SeedTreePaths}. It uses an auxiliary function \texttt{compute\_seeds\_to\_publish} described in Alg.~\ref{alg:ComputeSeedToPublish}.
\begin{algorithm}[!ht]
\caption{$\texttt{compute\_seeds\_to\_publish}$}\label{alg:ComputeSeedToPublish}
\begin{algorithmic}[1] 
\Require A vector $\vect{f}$ of size $t$ and the \textit{Reference Tree} $\vect{x}$. 
\Ensure Modified \textit{Reference Tree} $\vect{x}$.
\For{$i=0;\ i<t;\ i=i+1$}
\State $\vecentry{x}{2l-1+i} =\vecentry{f}{i}$
\EndFor

\For{$i = 2l-2;\ i\geq 0;\ i=i-1$} 
    \State $\vecentry{x}{i}=\vecentry{x}{2i+1}\vee\vecentry{x}{2i+2}$
\EndFor
\State return $\vect{x}$
\end{algorithmic}
\end{algorithm}
In the \texttt{SeedTreePaths} procedure, a tree $\vect{x}$ of size $4l-1$ is initialized with all zero. We call this tree as \textit{Reference Tree}. In Alg.~\ref{alg:ComputeSeedToPublish}, the values of the leaf nodes of the \textit{Reference Tree} are updated according to the $\vect{f}$ \textit{i.e.}, $\vecentry{x}{2l-1+i}$ are assigned the value $\vecentry{f}{i}$ for all $0\leq i<t$. The remaining nodes of the \textit{Reference Tree} are assigned the value using the formula $\vecentry{x}{i} = \vecentry{x}{2i + 1} \vee \vecentry{x}{2i + 2}$, signifying that if either child has a value of $1$, the corresponding parent will be assigned $1$. In this way, the value of the \textit{Reference Tree} $\vect{x}$ has been updated in a bottom-up approach. Now, some locations in \textit{Seed Tree} are to be published as $\vect{TreeNode}$ with the help of the \textit{Reference Tree}. Alg.~\ref{alg:SeedTreePaths} checks if the $i$-th node of \textit{Reference Tree} $\vecentry{x}{i}$ is zero and its parent $\vecentry{x}{Parent(i)}$ is 1, where the function $Parent(\cdot)$ is defined as follows:
$$Parent(i)=
\begin{cases}
    0&\text{if } i=0\\
    \lfloor\frac{i-1}{2}\rfloor&\text{otherwise}
\end{cases}
$$
If the validity check is satisfied, then $\vecentry{seed}{i}$, the $i$-th node of the \textit{Seed Tree} is appended to $\vect{TreeNode}$. 

\begin{example}\label{example:1}
    In Fig.~\ref{fig:NormalTree}, we have given an example for {leaf nodes} $2l=8$ and the vector $\vect{d}$ is chosen as $(0,\ 3,\ 1,\ 1,\ 0,\ 0,\ 0,\ 0)$. Then, the vector $\vect{f}$ will be $(0,\ 1,\ 1,\ 1,\ 0,\ 0,\ 0,\ 0)$. From Fig.~\ref{fig:NormalTree}, we can see that for $i=2,\ 7$ the condition "$\vecentry{x}{i}=0$ and $\vecentry{x}{Parent(i)}=1$" is satisfied. Therefore, the vector $\vect{TreeNode}=(\vecentry{seed}{2},\ \vecentry{seed}{7})$ and $\vect{rsp}=(\mat{Q}_{\vecentry{d}{1}}^{T}\overline{\mat{Q}}_{1},\ \mat{Q}_{\vecentry{d}{2}}^{T}\overline{\mat{Q}}_{2},\ \mat{Q}_{\vecentry{d}{3}}^{T}\overline{\mat{Q}}_{3})=(\mat{Q}_{3}^{T}\overline{\mat{Q}}_{1},\ \mat{Q}_{1}^{T}\overline{\mat{Q}}_{2},\ \mat{Q}_{1}^{T}\overline{\mat{Q}}_{3})$ will be extracted from \textit{Seed Tree} is revealed at the end. {From the seeds $\vecentry{seed}{2}$ $\vecentry{seed}{7}$, we can compute the leaf seeds $\vecentry{seed}{7},\ \vecentry{seed}{11},\ \vecentry{seed}{12},\ \vecentry{seed}{13},\ \vecentry{seed}{14}$ which are equals to the leaf seeds $\vecentry{ESEED}{0}$, $\vecentry{ESEED}{4},$ $\vecentry{ESEED}{5},$ $\vecentry{ESEED}{6},$ $\vecentry{ESEED}{7}$ respectively. From these seeds, we can compute the matrices $\widetilde{\mat{Q}_{0}},\ \widetilde{\mat{Q}_{4}},\ \widetilde{\mat{Q}_{5}},\ \widetilde{\mat{Q}_{6}},\ \widetilde{\mat{Q}_{7}}$. From the output of \texttt{LESS\_Sign} algorithm, we will get either $\widetilde{\mat{Q}_{j}}$ or $\mat{Q}_{\vecentry{d}{j}}^{T}\overline{\mat{Q}_{j}}$.}    
    \begin{figure}[h]
    \centering
    \includegraphics[width=.8\linewidth]{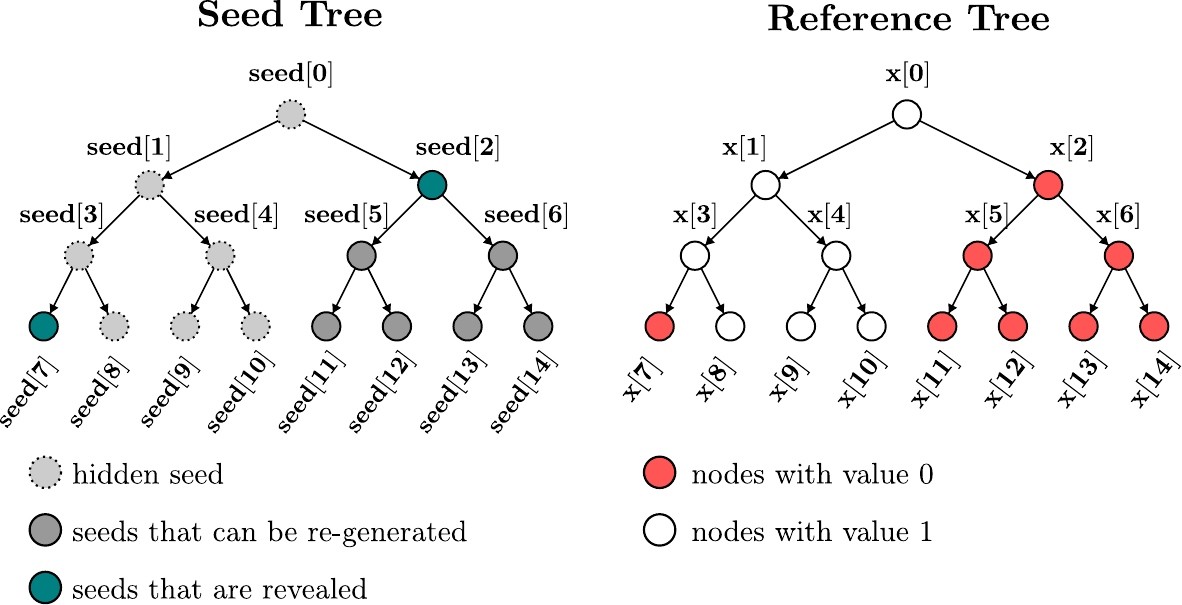}
    \caption{Example for extraction of $\vect{TreeNode}$ from Seed Tree using Reference Tree, following Alg.~\ref{alg:SeedTreePaths}} 
    \label{fig:NormalTree}  
    \end{figure}
\end{example}
As we can observe from Alg.~\ref{alg:SeedTreePaths}, the seeds from \textit{Seed Tree} that are revealed as $\vect{TreeNode}$ are directly associated with the values in $\vect{f}$ and the \textit{Reference Tree} $\vect{x}$. Therefore, we can try injecting faults in various locations of $\vect{f}$ or $\vect{x}$.

\subsubsection{\textbf{Modification of any node of the Reference Tree $\vect{x}$: }}\label{subsubsec:Modification of x} Here, we investigate the effect of modification of some fixed $i$-th value of the tree $\vect{x}$ in Alg.~\ref{alg:ComputeSeedToPublish}. Without loss of generality, assume $\vecentry{x}{i_{0}},\ \vecentry{x}{i_{1}},\ \cdots,\ \vecentry{x}{i_{r-1}}$ be the leaf nodes of the subtree with root node $\vecentry{x}{i}$, where $r\geq 1$. Now, suppose we {inject a fault} in the signature algorithm to modify the value of $i$-th node of the \textit{Reference Tree} $\vect{x}$. In that case, the signature algorithm will give us the faulted signature. However, even if we try to inject a fault in a physical machine, the fault can only occur with a certain probability. If we assume that the fault injection is successful, even then, there are several cases:
\begin{itemize}
    \item \textbf{Case 1: }The node $\vecentry{x}{i}$ is 0 in the non-faulted case. In this case, since the actual value $\vecentry{x}{i}$ is $0$, all the leaf nodes in the subtree rooted at $\vecentry{x}{i}$ must be zero. Hence, the vectors $\vect{f}$ and $\vect{d}$ do not have any non-zero value at the positions corresponding to the leaf nodes $\vecentry{x}{i_{0}},\ \vecentry{x}{i_{1}},\ \cdots,\ \vecentry{x}{i_{r-1}}$. Therefore, the $\vect{rsp}$ does not contain multiplication of any secret monomial matrix with the partial monomial matrix $\overline{\mat{Q}}_{i_{j}-2l+1}$, where $0\leq j<r$ \textit{i.e.}, we can not get any information about the secret matrices. 
    \item \textbf{Case 2: } The node $\vecentry{x}{i}$ is 1 in the non-faulted case, and after the fault injection, it has changed to 0. Since the \textit{Reference Tree} is updated in a bottom-up approach, the modification of the $i$-th node $\vecentry{x}{i}$ may affect the ancestors of $\vecentry{x}{i}$. Consequently, it may change the root node $\vecentry{x}{0}$. In this case, assume that it changes the value of the root node $\vecentry{x}{0}$ to 0. This case can occur only if all non-zero leaves fall under the subtree rooted at $\vecentry{x}{i}$. Since the value of the root node is zero, all the ancestors of $\vecentry{x}{i}$ including $\vecentry{x}{0}$ are zero. Therefore, neither $\vecentry{seed}{i}$ nor any of its ancestors in \textit{Seed Tree} is released because the Alg.~\ref{alg:SeedTreePaths} requires the parent of $\vecentry{x}{j}$ to be 1 if we want to release the $\vecentry{seed}{j}$, \textit{i.e.} such fault does not provide any advantage to us. Therefore, the nodes in the subtree rooted at $\vecentry{x}{i}$ do not affect the fault injection, so no extra information can be achieved from the released seeds corresponding to this subtree. 
    \begin{example}
  We consider the fixed digest vector $\vect{d}$, $\vect{f}$, and the \textit{Reference Tree} $\vect{x}$ of Example~\ref{example:1} in a non-faulted scenario. We modify the value of $\vecentry{x}{1}$ from $1\rightarrow 0$ that changes the value of $\vecentry{x}{0}$ from $1\rightarrow 0$. Fig.~\ref{fig:modify_x_Case2} represents the \textit{Reference Tree} and the related node of \textit{Seed Tree} in faulted case. In this case, only $\vecentry{x}{7}$ satisfies the condition "$\vecentry{x}{7}=0 \text{ and }\vecentry{x}{Parent(7)}=1$". Therefore, $\vect{TreeNode}$ will be $(\vecentry{seed}{7})$ and $\vect{rsp}=(\mat{Q}_{\vecentry{d}{1}}^{T}\overline{\mat{Q}}_{1},\ \mat{Q}_{\vecentry{d}{2}}^{T}\overline{\mat{Q}}_{2},\ \mat{Q}_{\vecentry{d}{3}}^{T}\overline{\mat{Q}}_{3})=(\mat{Q}_{3}^{T}\overline{\mat{Q}}_{1},\ \mat{Q}_{1}^{T}\overline{\mat{Q}}_{2},\ \mat{Q}_{1}^{T}\overline{\mat{Q}}_{3})$. 
  None of the monomial matrices $\widetilde{\mat{Q}}_{1},\ \widetilde{\mat{Q}}_{2},\ \widetilde{\mat{Q}}_{3}$ can be generated from $\vecentry{seed}{7}$. Therefore, we are unable to recover any secret key-related information from this faulted signature.
\end{example}
\begin{figure}[h]
\centering
\includegraphics[width=.8\linewidth]{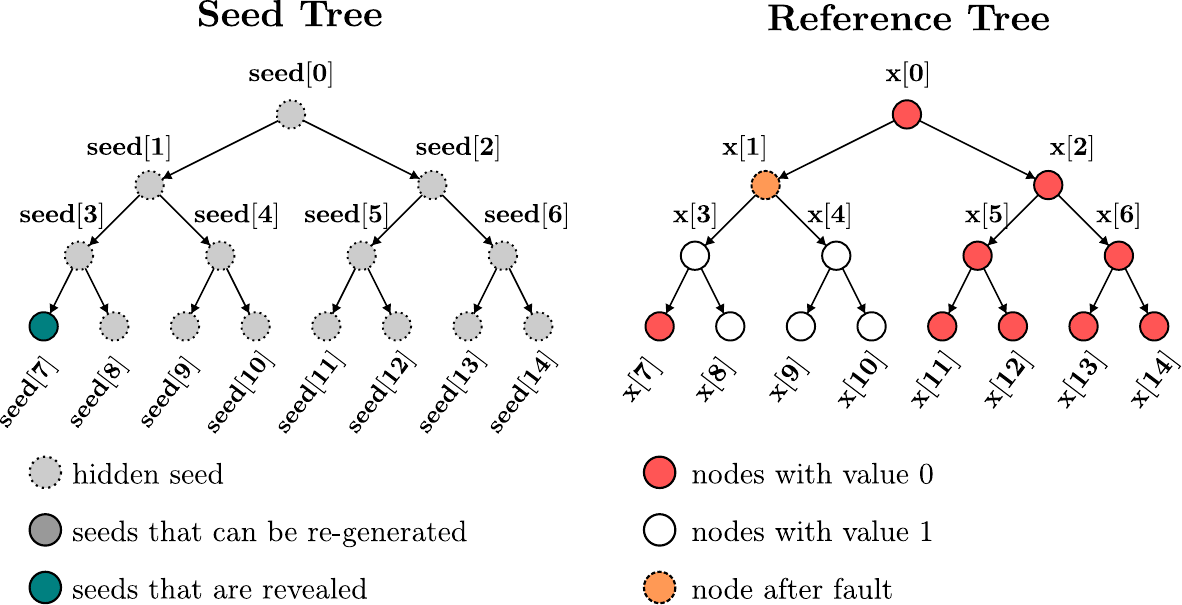}
\caption{Example of Case 2} 
\label{fig:modify_x_Case2}  
\end{figure}
     \item\textbf{Case 3: } The node $\vecentry{x}{i}$ is 1 in the non-faulted case. After the fault injection, it has changed to 0, but $\vecentry{x}{0}$ remains 1. Since the actual value of $\vecentry{x}{i}$ is $1$, there exists some leaf node $\vecentry{x}{i_{j}}$ such that $\vecentry{x}{i_{j}}=1$. Therefore, it follows that $\vecentry{f}{i_{j}-2l+1}$ is non-zero and consequently $\vecentry{d}{i_{j}-2l+1}$ is also non-zero. Without loss of generality, assume that $i_{j}-2l+1=k'$ then $\vect{rsp}$ contains $\mat{Q}_{\vecentry{d}{k'}}^{T}\overline{\mat{Q}}_{k'}$. Also, since the faulted value of $\vecentry{x}{i}$ is $0$ and $\vecentry{x}{0}=1$, so $\vect{TreeNode}$ will contain either $\vecentry{seed}{i}$ or any of its ancestors in \textit{Seed Tree} from which we can generate the leaf node $\vecentry{seed}{i_{j}}$ of the subtree rooted at $\vecentry{seed}{i}$. Hence, the ephemeral key $\vecentry{ESEED}{k'}$ and consequently the monomial matrix $\widetilde{\mat{Q}}_{k'}$ can be generated. Therefore, we retrieve the pair $(\widetilde{\mat{Q}}_{k'},\ \mat{Q}_{\vecentry{d}{k'}}^{T}\overline{\mat{Q}}_{k'})\,.$
 \begin{example}
  We consider the fixed digest vector $\vect{d}$, $\vect{f}$, and the \textit{Reference Tree} $\vect{x}$ of Example~\ref{example:1} in a non-faulted scenario. We modify the value of $\vecentry{x}{3}$ from $1\rightarrow 0$ that does not change the value of $\vecentry{x}{0}$. Fig.~\ref{fig:modify_x_Case3} represents the \textit{Reference Tree} and the related node of \textit{Seed Tree} in the faulted case. From this Fig.~\ref{fig:modify_x_Case3} we can see that for $i=2,\ 3$ the condition "$\vecentry{x}{i}=0$ and $\vecentry{x}{Parent(i)}=1$" is satisfied. Therefore, $\vect{TreeNode}$ will be $(\vecentry{seed}{2},\ \vecentry{seed}{3})$ and $\vect{rsp}=(\mat{Q}_{3}^{T}\overline{\mat{Q}}_{1},\ \mat{Q}_{1}^{T}\overline{\mat{Q}}_{2},\ \mat{Q}_{1}^{T}\overline{\mat{Q}}_{3})$. Now $\vecentry{seed}{3}$ is contained in the signature component \textit{Reference Tree} that we can generate the seed $\vecentry{seed}{8}$   to the $8-2l+1=8-8+1=1$-st monomial matrix $\widetilde{\mat{Q}}_{1}$. So, from this faulted signature, we found the pair $(\widetilde{\mat{Q}}_{1},\ \mat{Q}_{3}^{T}\overline{\mat{Q}}_{1})$ that help us find the information of the matrix $\mat{Q}_{3}^{T}$.
\end{example}
\begin{figure}[h]
\centering
\includegraphics[width=.8\linewidth]{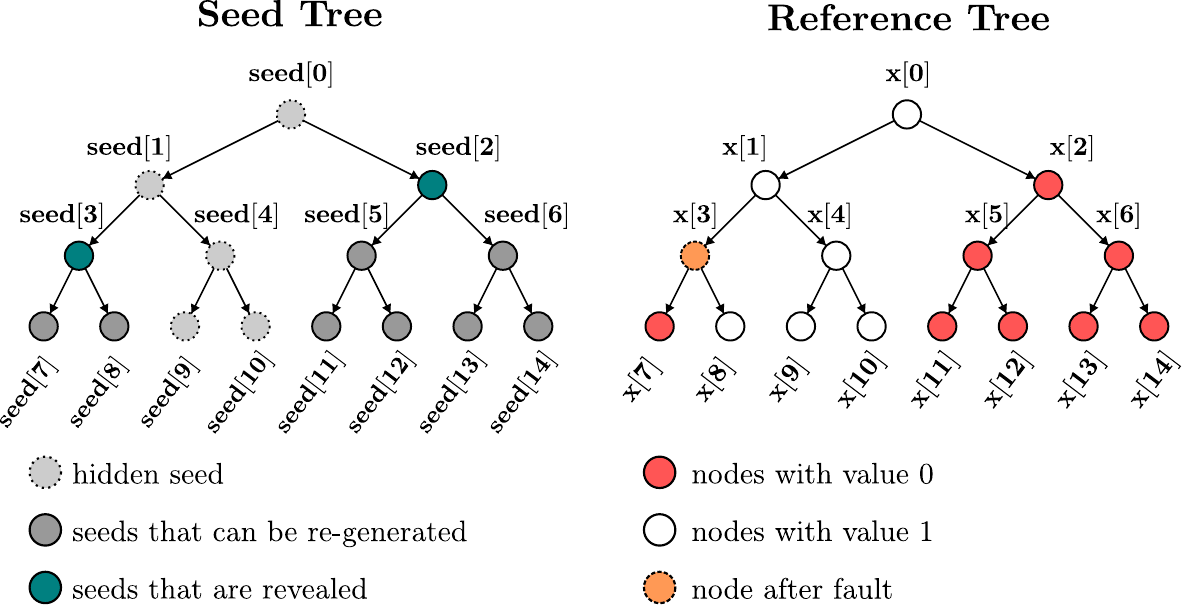}
\caption{Example of Case 3} 
\label{fig:modify_x_Case3}  
\end{figure}
\end{itemize}

\subsubsection{\textbf{Modification of the vector $\vect{f}$:} }
The vector $\vect{f}$ is computed by using the fixed digest vector $\vect{d}$. If the $i$-th element of $\vect{d}$ holds a non-zero value, $\vecentry{f}{i}$ is assigned the value of $1$; otherwise, it is set to zero. If we modify the $i$-th value $\vecentry{f}{i}$ by injecting fault, then the $(2l-1+i)$-th leaf node $\vecentry{x}{2l-1+i}$ of the \textit{Reference Tree} will be changed. The effect of this fault will be the same as the above modification of any leaf node of the Reference Tree $\vect{x}$. 

From the above, we can observe that the attack surfaces are different as in the second attack component, we change the value of any $\vecentry{f}{i}$, and in the first attack component, we change the value of any $\vecentry{x}{i}$. However, we can say that modifying the value of any value $\vecentry{f}{i}$ is imposing the same effect as modifying the corresponding leaf node of $\vecentry{x}{2l-1+i}$. Therefore, from now onwards, we only discuss the modification of any node of the \textit{Reference Tree} $ \vect{x}$.


\subsection{Fault models}\label{sec:FaultModel}
In this section, we describe the fault models that will help us to recover the secret key. {Our attack just requires changing a bit ($1\rightarrow0$ for LESS). Here we discussed in detailed how each fault model can be utilize to realize our attack}. 
Mainly, our fault assumptions can be realized by "skipping one condition check" or "forcing one data corruption in $\vect{f}$ or $\vect{x}$". We assume that the faulted location is arbitrary but known to the attacker. 
\subsubsection{{Skip the validity check condition in Alg.~\ref{alg:SeedTreePaths}: }} If we skip the check "$\vecentry{x}{i}=0\text{ and }\vecentry{x}{Parent(i)}=1$" in Alg.~\ref{alg:SeedTreePaths} for a fixed $i$, then $\vecentry{seed}{i}$ will always contained in $\vect{TreeNode}$. Without loss of generality, let $\vecentry{seed}{i_{0}},\ \cdots,\ \vecentry{seed}{i_{r-1}}$ be the leaf nodes of the subtree rooted the node $\vecentry{seed}{i}$ of $\textit{Seed Tree}$ and $\vecentry{x}{i_{0}},\ \cdots,\ \vecentry{x}{i_{r-1}}$ be the corresponding leaf nodes of the subtree rooted the node $\vecentry{x}{i}$ of the \textit{Reference Tree}. If $\vecentry{x}{i}=1$, then there exists the leaf node say $\vecentry{x}{i_{j'}}$, where $0\leq {j'}<r$ and $\vecentry{x}{i_{j'}}=1$. This implies $\vecentry{f}{i_{j'}-2l+1}=1$ and so, $\vecentry{d}{i_{j'}-2l+1}$ must be non-zero. Therefore, $\vect{rsp}$ must contain the matrix multiplication $\mat{Q}_{\vecentry{d}{i_{j'}-2l+1}}^{T}\overline{\mat{Q}}_{i_{j'}-2l+1}$. Since $\vecentry{d}{i_{j'}-2l+1}$ is non-zero, we are not supposed to have the information about the ephemeral matrix $\widetilde{\mat{Q}}_{i_{j'}-2l+1}$ in non-faulted case. However, from the $\vecentry{seed}{i}$, we can generate all leaf nodes of the subtree rooted in this node. Now, $\vecentry{seed}{i_{j'}}$ is the $i_{j'}-2l+1$-th leaf node $\vecentry{ESEED}{i_{j'}-2l+1}$ of the \textit{Seed Tree}, and that helps us find the ephemeral monomial matrix $\widetilde{\mat{Q}}_{i_{j'}-2l+1}$. So, in this case, we can find some information of secret monomial matrix $\mat{Q}_{\vecentry{d}{i_j'-2l+1}}^{T}$ from the pair $(\widetilde{\mat{Q}}_{i_{j'}-2l+1},\ \mat{Q}_{\vecentry{d}{i_j'-2l+1}}^{T}\overline{\mat{Q}}_{i_{j'}-2l+1})$. This is a model that we can use to mount the attack.

Previously, many works~\cite{Fault_Attacks_on_CCA-secure_Lattice_KEMs,DBLP:journals/tches/BruinderinkP18} have shown that instruction skips can be easily done with clock glitches, and the fault happens with very high probability. Mainly, in these works they have skipped the condition check instructions and store instructions. Recently, Keita et al. in~\cite{Fault-Injection_Attacks_Against_NISTs_Post-Quantum_Cryptography_Round_3_KEM_Candidates} bypassed the validity check in the decapsulation procedure in post-quantum the key-encapsulation mechanism Kyber~\cite{kyber_specification}. However, one may argue that skipping the validity check is the most important part as if we can skip this validity check for $i=0$ in line-6 in Alg.~\ref{alg:SeedTreePaths}, then $\vecentry{seed}{0}$ will be revealed. Henceforth all $\widetilde{\mat{Q}}_i$'s would have been revealed. Therefore, one may want to protect this checking at any cost. In fact, the need to protect this validity check was previously noted by Oder et al.~\cite{Practical_CCA2-Secure_and_Masked_Ring-LWE_Implementation} for different post-quantum schemes (e.g. Kyber). Nevertheless, the data in $\vect{d}$ and \textit{Reference Tree} $\vect{x}$ can be corrupted by skipping the storing instruction and forcing the data not to change. 

\noindent\textbf{Skip the store instruction in Alg.~\ref{alg:ComputeSeedToPublish} to corrupt $\vect{x}$: }
All the nodes of the \textit{Reference Tree} $\vect{x}$ are initialized by zero. If we can skip the store instruction $\vecentry{x}{i}=\vecentry{x}{2i+1}\vee\vecentry{x}{2i+2}$ for any $i$, then vaule of $\vecentry{x}{i}$ will remain zero. However, if the value of $\vecentry{x}{i}$ is supposed to be 1 in the non-faulted case, then $\vecentry{x}{i}$ will be modified after injecting this store instruction fault. 
As we discussed earlier, we can find the information of the secret matrix if the fault changes the value of $\vecentry{x}{i}$ from $1$ to $ 0$ and $\vecentry{x}{0}$ remains $1$. A similar instruction skipping attack has been shown in \cite{Fault_Attacks_on_CCA-secure_Lattice_KEMs}. 

\noindent\textbf{Stuck-at-zero fault model to corrupt $\vect{x}$: }
 A possible attack avenue is exploiting effective faults in the stuck-at model, where an attacker can try to alter the $i$-th intermediate value $\vecentry{x}{i}$ to a particular known value, e.g., to zero using stuck-at-zero fault~\cite{Secret_External_Encodings_Do_Not_Prevent_Transient_Fault_Analysis,DBLP:journals/iacr/GenetKPM18,DBLP:journals/iacr/KunduCSKMV23} using voltage glitches or electromagnetic attacks. The effect of this fault is equivalent to the above "store instruction skipping" fault. So, this fault will allow us to find the secret matrix.
 
 \noindent\textbf{Rowhammer attack model to corrupt $\vect{x}$: }Rowhammer \cite{A_new_approach_for_rowhammer_attacks} is a hardware bug identified in DRAMS (dynamic random access memory), where repeated row activations can cause bitflips in adjacent rows. This can also be a possible attack where bitflips ($1\to 0$) can be employed to corrupt the data $\vecentry{x}{i}$. Recently, such an attack on Kyber using rowhammer has been shown in~\cite{A_practical_key-recovery_attack_on_LWE-based_key-encapsulation_mechanism_schemes_using_Rowhammer}.\\

  As we discussed, any one of the above fault models can generate effective faulted signatures. However, the definition of successful fault always depends on the fault model. For example, if we work on the first fault model, $i.e.,$ "Skip the validity check condition in Alg.~\ref{alg:SeedTreePaths}", then the successful-fault will be: successfully skipped the checking condition "$\vecentry{x}{i}=0\text{ and }\vecentry{x}{Parent(i)}=1$" for a known $i$. But, if we work on the second fault model, then the successful fault will be: successfully skipped the store instruction "$\vecentry{x}{i}=\vecentry{x}{2i+1}\vee\vecentry{x}{2i+2}$" for a known $i$. 

From now on, we will discuss the second fault model to inject a fault, $i.e.,$ we inject a fault to skip the $i$-th store instruction 
$Ins(i):$ "$\vecentry{x}{i}=\vecentry{x}{2i+1}\vee\vecentry{x}{2i+2}$" for a fixed known $i$. Since all the values of the \textit{Reference Tree} $\vect{x}$ are initialized by zero, therefore for each successful fault, the value of $\vecentry{x}{i}$ will always be zero, where the position of fault location $\vecentry{x}{i}$ is known to the attacker. In the practical setup of this store instruction skip fault model, the following cases may arise: 
\begin{itemize}
    \item Successfully skipped the instruction $Ins(i)$, for the known $i$ and outputs the signature we call it a successful faulted signature. This fault could be an effective or ineffective fault.
    \item Could not skip the instruction $Ins(i)$, for the known fixed $i$. In this case, we call the output signature an unsuccessful faulted signature. 
\end{itemize}
In a physical device, faults can be induced with varying success rates. Even if there is a successful fault, the resulting faulty signature may or may not provide information about the secret key, as we observed earlier. A successful fault is called "effective" if it reveals secret key information and "ineffective" if it does not. More explicitly, we will say that a successful fault is effective if the fault changes the value of $\vecentry{x}{i}$ from $1$ to $0$, but the value of the root node $\vecentry{x}{0}$ remains unchanged, (i.e., $1$). Otherwise, the fault will be ineffective. We must identify the effective faulted signature from the received signature to find the errorless secret matrix. In the next Section~\ref{sec:Faultdetection}, we will discuss the effective faulted signature detection method.

\subsection{Effective fault detection}\label{sec:Faultdetection}
Let $\tau' = (\textit{salt},\ \textit{cmt},\ \vect{TreeNode'},\ \vect{rsp})$ be the received signature corresponding to the message $m$. We need to detect if the signature is generated from an "effective" or "ineffective" fault. The injected fault only affects the \textit{Reference Tree} $\vect{x}$, so the signature components $\textit{salt}$, $\textit{cmt}$ and $\vect{rsp}$ remain the same with the corresponding non-faulted signature components. We can compute the fixed digest vector $\vect{d}$ corresponding to the signature $\tau'$ from $\textit{cmt}$. From the fixed digest vector $\vect{d}$, we can compute the vector $\vect{f}$ and the \textit{Reference Tree} $\vect{x}$ for the non-faulted case. However, we can compute the successful faulted \textit{Reference Tree} $\vect{x'}$ from the \textit{Reference Tree} $\vect{x}$ by assigning the value of $\vecentry{x'}{i}=0$ and updating the ancestors of $\vecentry{x'}{i}$ accordingly, $i.e.,$ $\vect{x'}$ should be the \textit{Reference Tree} if the instruction $Ins(i)$ is skipped. Then, we will distinguish the "effective" faulted signature and "ineffective" faulted signature with the following process:
\begin{itemize}
    \item \textbf{Step-1:} First, we will check whether $\vecentry{x}{i}=0$ or not. If $\vecentry{x}{i}=0$, then this is already a case of "ineffective fault", and we reject the signature. Otherwise, we will go to the next step.
    \item \textbf{Step-2:} Next, we will check whether $\vecentry{x'}{0}=0$ or not. If $\vecentry{x'}{0}=0$, then this is a case of "ineffective fault", and we reject the signature. Otherwise, from the \textit{Reference Tree} $\vect{x'}$, we compute the size of successfully faulted $\vect{TreeNode'}$, say $\Delta_{exp}$ and the size of received $\vect{TreeNode}'$ say $\Delta_{rec}$. We compare the values $\Delta_{exp}$ with $\Delta_{rec}$. 

    \item \textbf{Step-3:} If $\Delta_{rec} \neq \Delta_{exp}$, then the fault is unsuccessful, and we reject the signature. Otherwise, using $salt$, $\vect{TreeNode}'$ and $\vect{x}'$ we compute all the $\widetilde{\mat{Q}}_j$ where $\vecentry{d}{j}=0$. We apply the verification using these $\widetilde{\mat{Q}}_j$'s and $\vec{rsp}$. If the verification is successful, then we take the received signature as an effective faulted signature. Otherwise, we reject the signature.
    
    
\end{itemize}
{Note that, in \textbf{Step-3} of the above process, $\vecentry{x}{i}$ is changed from $1\to 0$ and $\vecentry{x}{0}=1$, but we still consider it as "unsuccessful" fault. This is because we want our fault detection method to detect whether our fault has been successfully injected exactly at the $i$-th location or not. If $\vecentry{x}{i}$ changed from $1\to 0$, then there are two cases.} 
\begin{itemize}
    \item \textit{Case-1: } {fault was successfully injected at $i$-th location.}
    \item \textit{Case-2: } {fault was injected at a $j$-th location for $j\neq i$ and it has changed $\vecentry{x}{i}$.}
\end{itemize}
{ We only consider \textit{Case-1} as “successful-fault”, but not \textit{Case-2} as the fault is not injected at the $i$-th location in that case.} In this procedure, we can detect that the faulted signature that is generated by successfully skipping the instruction 
$Ins(i)$ and that leaks the information about the secret matrix. Note that the targeted faulted location $i$ is arbitrary but known to the attacker. For simplicity, we will fix the targeted fault position $i$. We will check whether this fixed $i$-th store instruction $Ins(i)$ skipped and that leaks the information about the secret matrix or not. If this detection method passes, then we will use this signature. Otherwise, we will again query for another signature. 

\subsection{Attack template}\label{sec:attackTemplate}
In this section, first, we will describe how to obtain the secret monomial matrices from an effective faulted signature $\tau=(salt,\ cmt,\ \vect{Tree Node},\ \vect{rsp})$ in Alg.~\ref{alg:RecoverSecret}. 
Let $\vecentry{x}{i}$ be the node in \textit{Reference Tree} with height $h$, and $L_{\vecentry{x}{i}}$ be the set of all leaf nodes of the subtree rooted at $\vecentry{x}{i}$. We only need the leaf nodes from $L_{\vecentry{x}{i}}$ that coincide with the first $t$ (the length of the digest $\vect{d}$) many leaf nodes of the full \textit{Reference Tree}. Without loss of generality, assume that there are $v$ many such leaves, and let the set of indices of these leaves be 
$I^{(i)}_{\text{leaf}} = \{j_1,\ j_2,\ \cdots,\ j_v\}\,.$
From this effective faulted signature, all the secret matrices $\mat{Q}_{\vecentry{d}{j-2l+1}}^{T}$  will be recovered with Alg.~\ref{alg:RecoverSecret}, where $j\in I_{\text{leaf}}^{(i)}$ and $\vecentry{d}{j-2l+1}\neq 0$.
\begin{algorithm}
\caption{\texttt{Recover\_Secret\_Matrices}($\tau,\ PK$)}\label{alg:RecoverSecret}
\begin{algorithmic}[1]
\Require Signature $\tau=(\textit{salt},\ \textit{cmt},\ \vect{TreeNode},\ \vect{rsp})$, public key $PK=(gseed,\ \mat{G}_{1},\ \cdots,\ \mat{G}_{s-1})$.
\Ensure The columns of secret matrices $\mat{Q}_{\vecentry{d}{j-2l+1}}^{T}$, where $j\in I_{\text{leaf}}^{(i)}$ and $\vecentry{d}{j-2l+1}\neq 0$.

\State $\vect{d}\leftarrow \texttt{CSPRNG}(cmt,\ \mathbb{S}_{t,w})$
\State $\vect{seed}\leftarrow \texttt{SeedTreeUpdate}(\vect{TreeNode},\ salt, \ \vect{d})$
\State $\vect{ESEED}\leftarrow $ Leaf nodes of $\vect{seed}$ corresponding to $\vecentry{seed}{i}$
\State $\mat{G}_{0}\leftarrow \texttt{CSPRNG}(gseed,\ \mathbb{S}_{\texttt{RREF}})$

\For{$r=1;\ r\leq v;\ r=r+1$}
    \If{$\vecentry{d}{j_{r}-2l+1}\neq 0$ {\textbf{and} $\mat{Q}_{\vecentry{d}{j_{r}-2l+1}}$ is not recovered}}
       \State $\widetilde{\mat{Q}}_{j_{r}-2l+1}\leftarrow \texttt{CSPRNG}(\vecentry{ESEED}{j_{r}-2l+1},\ {M}_{n}(q))$
       \State $(\overline{\mat{Q}}_{j_{r}-2l+1},\ \overline{\mat{V}}_{j_{r}-2l+1})\leftarrow \texttt{PrepareDigestInput}(\mat{G}_{0},\ \widetilde{\mat{Q}}_{j_{r}-2l+1})$
       \State $\mat{Q}^{*}=\mat{Q}_{\vecentry{d}{j_{r}-2l+1}}^{T}\overline{\mat{Q}}_{j_{r}-2l+1}\leftarrow \texttt{ExpandToMonomAction}(\vect{rsp})$ 
       \State Compute $\mat{Q}_{\vecentry{d}{j_{r}-2l+1}}^{T}~\text{from}~(\mat{Q}^{*},\ \overline{\mat{Q}}_{j_{r}-2l+1})$ \Comment{following Section~\ref{subsec:Observation}}
        \EndIf
\EndFor
\end{algorithmic}
\end{algorithm}

Here, \texttt{SeedTreeUpdate} function takes the $\vect{TreeNode}$, \textit{salt} and digest $\vect{d}$ and generates all the ephemeral seeds assuming the modified \textit{Reference Tree} after effective fault. Since $\vecentry{seed}{i}$ is revealed in $\vect{TreeNode}$ after effective fault, we can say that $\vecentry{ESEED}{j-2l+1}$ for all $j\in I_{\text{leaf}}^{(i)}$ are revealed.

In this attack model, we are able to get into the victim's device and introduce the fault that causes it to bypass the $Ins(i)$ instruction. The \texttt{LESS\_KeyGen} (Alg.~\ref{alg:KeyGen}) is a one-time operation from where the secret key $\text{SK}=( MSSED,\ gseed)$ and public key $\text{PK}=(gseed,\ \mat{G}_{1},\ \cdots,\ \mat{G}_{s-1})$ are generated. But with this private key $\text{SK}$, the \texttt{LESS\_Sign} (Alg.~\ref{alg:Signature}) can execute more than once. We follow the following subsequent actions to find the secret monomial matrices:
\begin{itemize}
    \item \textbf{Step-1:} We generate a message, signature pair $(m,\ \tau)$ from the victim device.
    \item \textbf{Step-2:} After receiving the pair $(m,\ \tau)$, we will determine whether or not $\tau$ is an effective faulted signature. Go back to \textbf{Step-1} if the signature is not effective. If yes, then go to \textbf{Step-3}.
    \item \textbf{Step-3:} Using this signature $\tau$, we will run the \texttt{Recover\_Secret\_Matrices} algorithm (Alg.~\ref{alg:RecoverSecret}) to determine the hidden monomial matrices.
    \item \textbf{Step-4:} Next, we will calculate whether or not the whole secret monomial matrices were obtained. We terminate the process if the secret matrices are recovered. Otherwise, we repeat the same procedure to obtain the remaining non-recovered columns. 
\end{itemize}

\subsection{{Secret recovery from single fault}}\label{subsec:RecoveryColCalculation}
{In this section, we calculate the expected number of secret monomials recovered from one effective faulted signature where the fault is injected at a node $\vecentry{x}{i}$ ($0\leq i \leq 4l-2$). Now, if there are $m$ many non-zero leaves with distinct values in the subtree rooted at $\vecentry{x}{i}$, then we will get exactly $m$ many pairs of the form $(\widetilde{\mat{Q}}_{{j}},\ \mat{Q}_{\vecentry{d}{j}}^{T}\overline{\mat{Q}}_{{j}})$ $i.e.$ we recover $m$ many secret monomials. In this section, first, we will estimate the value of $m$. 

Suppose $L_{\vecentry{x}{i}}$ the set of leaf nodes in the subtree rooted at $\vecentry{x}{i}$ and let $|L_{\vecentry{x}{i}}|=\ell$. Let $W$ be the random variable representing the number of leaf nodes in $L_{\vecentry{x}{i}}$ with non-zero value. $X$ be a random variable that represents the number of distinct non-zero values of the leaf nodes in $L_{\vecentry{x}{i}}$. Then for any $0\leq m\leq s-1$, we have
}
\begin{align*}
    \Pr[X=m] = \sum_{r=m}^w\Pr[X=m~|~W=r]\cdot\Pr[W=r]
\end{align*}
{
Where $w$ is the weight of $\vect{d}$ and therefore $L_{\vecentry{x}{i}}$ can only have at most $w$ many non-zero valued leaf nodes. First, we will calculate $\Pr[X=m~|~W=r]$, which is the probability that $r$ many non-zero leaves take exactly $m$ many distinct values. These $m$ distinct values can be chosen from $(s-1)$ possible values in $\binom{s-1}{m}$ ways. Now, we have to assign all these $m$ values to the $r$ many leaf nodes. We first partition the $r$ locations into $m$ many non-empty subsets, which can be done in $S(r,\ m)$ many ways. This $S(r,\ m)$ is a \textit{Stirling number of the second kind} {\cite{On_stirling_numbers_of_the_second_kind}}. Now, each of the $m$ many subsets can be assigned a unique non-zero value, which can be done in $m!$ ways. So, the $r$ many leaf nodes can be assigned $m$ distinct value in $m!\binom{s-1}{m}S(r,\ m)$ ways. Therefore
}
\begin{align*}
    \Pr[X=m~|~W=r] =\frac{m!\binom{s-1}{m}S(r,\ m)}{(s-1)^r} 
\end{align*}
{Now, $\vect{d}$ has weight $w$ and $\Pr[W = r]$ is the probability that the $\ell$ many
locations of $\vect{d}$ corresponding to the leaf nodes in $L_{\vecentry{x}{i}}$ has exactly $r$ many non-zero values and the last $t-\ell$ many locations has $w-r$ many non-zero values. Therefore,}
\begin{align*}
    \Pr[W=r]=\frac{\binom{\ell}{r}\binom{t-\ell}{w-r}}{\binom{t}{w}}
\end{align*}
{Now we can calculate $\Pr[X=m]$ for all $0\leq m\leq s-1$. However, we are interested in finding the expected number of secret monomials with one single fault, which is the expectation of the random variable $X$.
}
\begin{align*}
    \Exp{X}&=\sum_{m=1}^{s-1}m\cdot\Pr[X=m]\\
    &=\sum_{m=1}^{s-1}m\left(\sum_{r=m}^{w}\frac{m!\binom{s-1}{m}S(r,\ m)}{(s-1)^r}\cdot\frac{\binom{\ell}{r}\binom{t-\ell}{w-r}}{\binom{t}{w}} \right)
\end{align*}
{With only one single faulted signature the expected number of secret monomials that we recover is $\Exp{X}$ but the total number of secret monomials is $(s-1)$. Therefore, we need multiple faulted signatures to recover all the secret monomials.
}

\section{Extending our attack to CROSS }\label{sec:CROSS}
\begin{algorithm}[!ht]
\caption{\texttt{CROSS\_Sign}
($Msg,\ \vect{e}$)}\label{alg:CROSS}
\begin{algorithmic}[1]
\Require Secret key $\vect{e}\in G$ and message $Msg$ where $G\subset \mathbb{E}^{n}$, $\mathbf{H}\in\mathbb{F}_{p}^{(n-k)\times n}$ are public key satisfying $\vect{s}=\vect{e}\mat{H}^{T}$ 
\Ensure Signature $\tau=\left\{Salt,\ c_{0},\ c_{1},\ h,\ \vect{SeedPath},\ \vect{MerkleProofs},\ \left\{f^{(i)}\right\}_{i\notin J}\right\}$
\State Sample $MSeed\xleftarrow[]{\$}\left\{0,\ 1\right\}^{\lambda},\ Salt\xleftarrow[]{\$}\left\{0,\ 1\right\}^{2\lambda}$ 
\State Generate $\vect{Seed}=\texttt{SeedTree}(MSeed,\ Salt)$ 
\State $\vecentry{ESEED}{1},\ \cdots,\ \vecentry{ESEED}{t}=$ Leaf nodes of $\vect{Seed}$ 
\For{$i=1,\ i\leq t,\ i=i+1$} 
\State Sample $(Seed^{(\vect{u'})},\ Seed^{(\vect{v'})})\xleftarrow[]{\vecentry{ESEED}{i}}\left\{0,\ 1\right\}^{2\lambda}$ 
\State Sample $\vect{u'}^{(i)}\xleftarrow[]{Seed^{(\vect{u'})}}\field{F}{p}^{n}$, $\vect{e'}^{(i)}\xleftarrow[]{Seed^{(\vect{e'})}}G$
\State Compute $\sigma^{(i)}\in G$ such that $\sigma^{(i)}(\vect{e'}^{(i)})=\vect{e}$
\State Set $\vect{u}^{(i)}=\sigma^{(i)}(\vect{u'}^{(i)})$
\State Compute $\widetilde{\vect{s}}^{(i)}=\vect{u}^{(i)}\mat{H}^{T}$
\State Set $c_{0}^{(i)}=\texttt{Hash}(\widetilde{\vect{s}}^{(i)},\ \sigma^{(i)},\ Salt,\ i)$
\State Set $c_{1}^{(i)}=\texttt{Hash}(\vect{u'}^{(i)},\ \vect{e'}^{(i)},\ Salt,\ i)$
\EndFor

\State Set $\mathcal{T}=\texttt{Merkle Tree}(c_{0}^{(1)},\ \cdots,\ c_{0}^{(t)})$ 
\State Compute $c_{0}=\mathcal{T}.\texttt{Root}()$ 
\State Compute $c_{1}=\texttt{Hash}(c_{1}^{(1)},\ \cdots,\ c_{1}^{(t)})$
\State Generate $(\beta^{(1)},\ \cdots,\ \beta^{(t)})=\texttt{GenCh}_{1}(c_{0},\ c_{1},\ Msg,\ Salt)$ 
\For{$i=1,\ i\leq t,\ i=i+1$} 
\State Compute $\vect{y}^{(i)}=\vect{u'}^{(i)}+\beta^{(i)}\vect{e'}^{(i)}$
\State Compute $h^{(i)}=\texttt{Hash}(\vect{y}^{(i)})$
\EndFor
\State Compute $h=\texttt{Hash}(h^{(1)},\ \cdots,\ h^{(t)})$
\State Generate $\left(\vecentry{b}{1},\ \cdots,\ \vecentry{b}{t}\right)=\texttt{GenCh}_{2}(c_{0},\ c_{1},\ \beta^{(1)},\ \cdots,\ \beta^{(t)},\ h,\ Msg,\ Salt)$
\State Set $J=\left\{i:~ \vecentry{b}{i}=1\right\}$
\State Set $\vect{SeedPath}=\texttt{publish\_seeds}(MSeed,\ Salt,\ J)$
\For{$i\notin J$}
\State $f^{(i)}:=(\vect{y}^{(i)},\ \sigma^{(i)},\ c_{1}^{(i)})$
\EndFor
\State Compute $\vect{MerkleProofs}=\mathcal{T}.\texttt{Proofs}(\left\{1,\ \cdots,\ t\right\}\setminus J)$
\State Return $\tau=\left\{Salt,\ c_{0},\ c_{1},\ h,\ \vect{SeedPath},\ \vect{MerkleProofs},\ \left\{f^{(i)}\right\}_{i\notin J}\right\}$ 
\end{algorithmic}
\end{algorithm}

\noindent CROSS uses $\mathbb{E}^{n}$ {a commutative group isomorphic to} $(\field{F}{z}^{n},\ +)$, where $n,\ z$ are parameters of the signature and $G$ is a subgroup of $\mathbb{E}^{n}$. Here $\vect{e}$ is a long-term secret vector which is used {to generate signatures}. Therefore, the attacker can generate multiple valid signatures using the secret $\vect{e}$ information. Similar to the attack on LESS, the target here is to find the information of the secret vector $\vect{e}$. In the Alg.~\ref{alg:CROSS}, we can observe that if we have information of one single pair $(\vect{e'}^{(i)},\ f^{(i)}=(\vect{y}^{(i)},\ \sigma^{(i)},\ c_{1}^{(i)}))$, then we can compute the secret $\vect{e}$ by $\vect{e}=\sigma^{(i)}(\vect{e'}^{(i)})$. Therefore, we aim to find one such pair corresponding to any $i$ for full key recovery. 

\begin{algorithm}[!ht]
\caption{\texttt{Recover\_Secret\_CROSS}
($\tau,\ PK$)}\label{alg:recover-secret-cross}
\begin{algorithmic}[1]
\Require $\tau=\left\{Salt,\ c_{0},\ c_{1},\ h,\ \vect{SeedPath},\ \vect{MerkleProofs},\ \left\{(\vect{y}^{(i)},\ \sigma^{(i)},\ c_{1}^{(i)})\right\}_{i\notin J}\right\}$
\Ensure The secret vector $\vect{e}$.
\State Generate $(\beta^{(1)},\ \cdots,\ \beta^{(t)})=\texttt{GenCh}_{1}(c_{0},\ c_{1},\ Msg,\ Salt)$
\State Generate $\left(\vecentry{b}{1},\ \cdots,\ \vecentry{b}{t}\right)=\texttt{GenCh}_{2}(c_{0},\ c_{1},\ \beta^{(1)},\ \cdots,\ \beta^{(t)},\ h,\ Msg,\ Salt)$
\State Set $J=\left\{i:~ \vecentry{b}{i}=1\right\}$
\State {$\vecentry{ESEED}{j_{1}-2l},\ \cdots,\ \vecentry{ESEED}{j_{v}-2l}\leftarrow \texttt{SeedTreeUpdate}(\vecentry{seed}{i},\ Salt, \ \vect{b})$}
 \For{{$i=j_{1}-2l,\ \cdots,\ j_{v}-2l:$}}
 \If{$i\notin J$}
\State Sample $(Seed^{(\vect{u'})},\ Seed^{(\vect{v'})})\xleftarrow[]{\vecentry{ESEED}{i}}\left\{0,\ 1\right\}^{2\lambda}$ 
\State Sample $\vect{e'}^{(i)}\xleftarrow[]{Seed^{(\vect{e'})}}G$
\State Compute $\vect{e}=\sigma^{(i)}(\vect{e'}^{(i)})$
\State\Return $\vect{e}$
\EndIf
\EndFor
\end{algorithmic}
\end{algorithm}

In Alg.~\ref{alg:CROSS}, the function \texttt{publish\_seeds} (line 22) works equivalent to the function \texttt{SeedTreePaths}~(Alg.~\ref{alg:SeedTreePaths}) used in LESS signature. Using the digest vector $\vect{b}$, the function \texttt{publish\_seeds} first creates a \textit{Reference Tree} say $\vect{y}$ in a bottom-up approach like LESS signature. The only difference in this \texttt{Reference Tree} $\vect{y}$ is that the flag of the published seed is defined as $1$ and unpublished seed notation as $0$, whereas in LESS, the authors define the opposite. However, Both use equivalent concepts. Like LESS, we require the modification of any node of the \textit{Reference Tree} from $1$(\text{ flag of the unpublished seed }) $\rightarrow 0$(\text{ flag of the published seed }) to get an effective faulted signature. Therefore, we need the modification from $0\to 1$ to get an effective faulted signature. We can detect the effective faulted signature here using a similar technique that we used in Section~\ref{sec:Faultdetection} to detect effective fault for LESS signature. 

Let us assume that we apply fault injection to the CROSS signature of a victim's device such that the value of {$\vecentry{y}{i}$} has been changed from $0\to 1$. Consider an effective faulted signature as
$\tau=\left\{Salt,\ c_{0},\ c_{1},\ h,\ \vect{SeedPath},\ \vect{MerkleProofs},\ \left\{f^{(i)}\right\}_{i\notin J}\right\}\,.$
{Since $\tau$ is an effective-faulted signature, therefore we will get the seed $\vecentry{seed}{i}$. All the leaf nodes of the subtree say $L_{\vecentry{x}{i}}=\left\{\vecentry{x}{j_{1}},\ \cdots,\ \vecentry{x}{j_{v}}\right\}$ rooted as $\vecentry{seed}{i}$ can be computed from $\vecentry{seed}{i}$. i.e., $\vecentry{ESEED}{j_{1}-2l},\ \cdots,\ \vecentry{ESEED}{j_{v}-2l}$ will be the corresponding leaf ephemeral seeds.} Now, we will find the secret key $\vect{e}$ using the Alg.~\ref{alg:recover-secret-cross}. The function \texttt{SeedTreeUpdate} works the same way we defined it in Section~\ref{sec:attackTemplate}.

\section{Simulation result}\label{section:simulation}

{In this section, we discuss the simulation procedure of our fault attack on LESS and CROSS signatures; $i.e.$, we apply our fault assumption inside the LESS and CROSS signature algorithms to imitate the corresponding practical attack scenario. The simulation code is available at GitHub \footnote{\label{link}\url{https://github.com/s-adhikary/zkfault_simulation}}.

In the previous Sections{~\ref{subsec:RecoveryColCalculation}} and {~\ref{sec:CROSS}}, we have analyzed the effect of modification of the values $\vecentry{x}{i}$ (for LESS) and $\vecentry{y}{i}$ (for CROSS) to $0$ and $1$ respectively. For LESS and CROSS, this can be achieved by stuck at zero/stuck at one or instruction skip fault. So, in the simulation code, we have assumed the values $0$ and $1$ of the nodes $\vecentry{x}{i}$ and $\vecentry{y}{i}$ respectively. After receiving this faulted signature $\tau$, we compute the corresponding secrets of CROSS (secret $\vect{e}$) and LESS (secret monomial matrices) with the help of the respective algorithms Alg.{~\ref{alg:RecoverSecret}} and Alg.{~\ref{alg:recover-secret-cross}}}.

{Note that the attack is valid if we target any $\vecentry{x}{i}$ ($\vecentry{y}{i}$) for fault injection in LESS (CROSS) signature, where $i$ is an arbitrary but fixed location. But in our simulation code we have fixed the location as $i=1$. One may change this location and the simulation code accordingly.  However, in that case the results in Table.~{\ref{tab:ExpectedCOFS}} would change according to our result in Section~{\ref{subsec:RecoveryColCalculation}}.}
{We provide the simulation results for all the versions of LESS and CROSS in Table{~\ref{tab:ExpectedCOFS}}. We have run the simulation code multiple times to recover all secrets with (multiple) faulted signatures. We take the average of the number of faulted signatures required to recover all secrets, which we denote with $N_{\text{avg}}$. We have also included the average number of secrets recovered from one single fault ($\Exp{X}$) in the table.}

\begin{table}[!ht]
\centering
\begin{tabular}{c||cccccc}
\hline
\bf Scheme ~&\begin{tabular}[x]{@{}c@{}}~\bf Security\\\bf Level\end{tabular}~&
\begin{tabular}[x]{@{}c@{}}\bf Parameter\\\bf Set\end{tabular}~&
\begin{tabular}[x]{@{}c@{}}\bf Optim.\\\bf Corner\end{tabular}~&
\begin{tabular}[x]{@{}c@{}}\bf Number\\\bf of Secrets\end{tabular}~&
$~\Exp{X}~$  &
\textbf{$N_{\text{avg}}$} \\ \hline
\multirow{7}{*}{LESS~\cite{LESS_is_More}} &\multirow{3}{*}{1}
                    & LESS-1b & -& 1 & 1 & 1\\
                  & & LESS-1i &-   & 3 & 2.91 & 1.05\\
                  & & LESS-1s & -&7  &5.55& 2.09\\ 
                   \cline{2-7}
&\multirow{2}{*}{3} & LESS-3b & -&1   &1& 1\\
                  & & LESS-3s &-& 2   &2&1\\
                    \cline{2-7}
&\multirow{2}{*}{5} & LESS-5b & -&1 &1 & 1\\
                  & & LESS-5s &-& 2 &2 &1\\
                    \hline
\multirow{2}{*}{CROSS~\cite{CROSS_Specification_Doc}}&\multirow{2}{*}{1,\ 3, \& 5} & CROSS-R-SDP& fast/small& 1&1&1\\
                  &  & CROSS-R-SDP(G)&fast/small& 1&1  &1\\\hline
\end{tabular}%
\caption{{Simulation result of full secret monomial matrices recovery of LESS and CROSS signature~{\cite{LESS_is_More,CROSS_Specification_Doc}}}}.
\label{tab:ExpectedCOFS}
\end{table}

{Our analysis is based on the fact that each time we query the faulted signature oracle, we get an effectively faulted signature. However, in a practical fault attack, this is not the case. In the real world, there is a probability that an injected fault is successful, say $p_1$, and also there is a probability that a successfully injected fault is effective, say $p_0$. Then}
\begin{align}
\begin{aligned}
	&\Pr[~\text{effective fault}~\wedge~\text{successful fault}~]\\
	=&\Pr[~\text{effective fault}~|~\text{successful fault}~]\cdot \Pr[~\text{successful fault}~]\\
	=& p_0p_1
\end{aligned}
\end{align}
{Let us consider $p=p_0p_1$. Therefore, in a practical scenario to get one faulted signature, the approximate number of queries to the faulted signature oracle needed would be $N_{\text{trial}}=\frac{1}{p}$. Moreover, to get $N_{\text{avg}}$ many faulted signatures, we need $N_{\text{total}}=\frac{N_{\text{avg}}}{p}$ many queries. For example if we consider $p=0.01$, then $N_{\text{trial}}=100$ and consequently, $N_{\text{total}}=100\cdot N_{\text{avg}}$.}

\section{Countermeasures}\label{sec:Countermeasure}
In the previous section, we have seen that the primary attack surface $\textit{Reference Tree}$ $\vect{x}$ is initialized by 0. If we inject fault to skip store instruction line-5 of Alg.~\ref{alg:ComputeSeedToPublish} i.e. $\vecentry{x}{i}=\vecentry{x}{2i+1}\wedge\vecentry{x}{2i+2}$, then $\vecentry{x}{i}$ does not change the value and stays 0. Hence, one may suggest initializing the \textit{Reference Tree} with all 1.
The instruction skip fault does not work in this case, but we can apply bit-flip fault or stuck-at-zero faults and apply the same attack analysis. Since many practical fault attacks are applicable, countermeasures against one type of fault may not serve our purpose. Therefore, first, we must identify the main reason for the existence of the attack vector.

After the digest computation in Alg.~\ref{alg:Signature}, the values of vector $\vect{d}$ are checked twice. First, by checking whether the value of each $\vecentry{d}{i}$ is zero or not, they published the component $\vect{TreeNode}$. Completing this procedure, again, each $\vecentry{d}{i}$ is checked to publish the component $\vect{rsp}$. Therefore, if an attacker injects a fault at the time of computing $\vect{TreeNode}$ and somehow succeeds in disclosing the seed $\vecentry{ESEED}{i}$ without altering the vector $\vect{d}$, then the information about the secret matrix $\mat{Q}_{\vecentry{d}{i}}^{T}$ is susceptible to leakage. To mitigate potential attacks, we must publish either the response $\widetilde{\mat{Q}}_{i}$ or $\mat{Q}_{\vecentry{d}{i}}^{T}\overline{\mat{Q}}_{i}$ after a single verification of the value $\vecentry{d}{i}$. 

In the following sections, we will offer concise explanations for two countermeasures incorporated within the LESS scheme that protect the scheme up to one fault.%
\subsection{Countermeasure with larger signature size}
The most straightforward countermeasure would be not using the tree construction at all. In this version, the preparation of digest $\vect{d}$ is the same as Alg.~\ref{alg:Signature}. After the digest preparation, for each $0\leq i<t$ we only check the value of $\vecentry{d}{i}$, and set 
$$\vecentry{rsp}{i}=\begin{cases}
    \widetilde{\mat{Q}}_i &\text{~~if~~}\vecentry{d}{i}=0\\
    \mat{Q}_{\vecentry{d}{i}}^T\overline{\mat{Q}}_i &\text{~~otherwise} 
\end{cases}\,.
$$
Then, we cannot get both $\widetilde{\mat{Q}}_{i}$ and $\mat{Q}_{\vecentry{d}{i}}^{T}\overline{\mat{Q}}_{i}$ for any $i$ and the attack can be prevented. However, in this case, the size of the signature will be $|\textit{cmt}|+wk(\lceil\log n\rceil+\lceil\log (q-1)\rceil)+(t-w)\lambda$. This signature size will be larger than the submitted version of LESS~\cite{LESS_Specification_Doc}. 
\subsection{Countermeasure with same small signature size}
\begin{figure}[b]
\centering
\includegraphics[width=.8\linewidth]{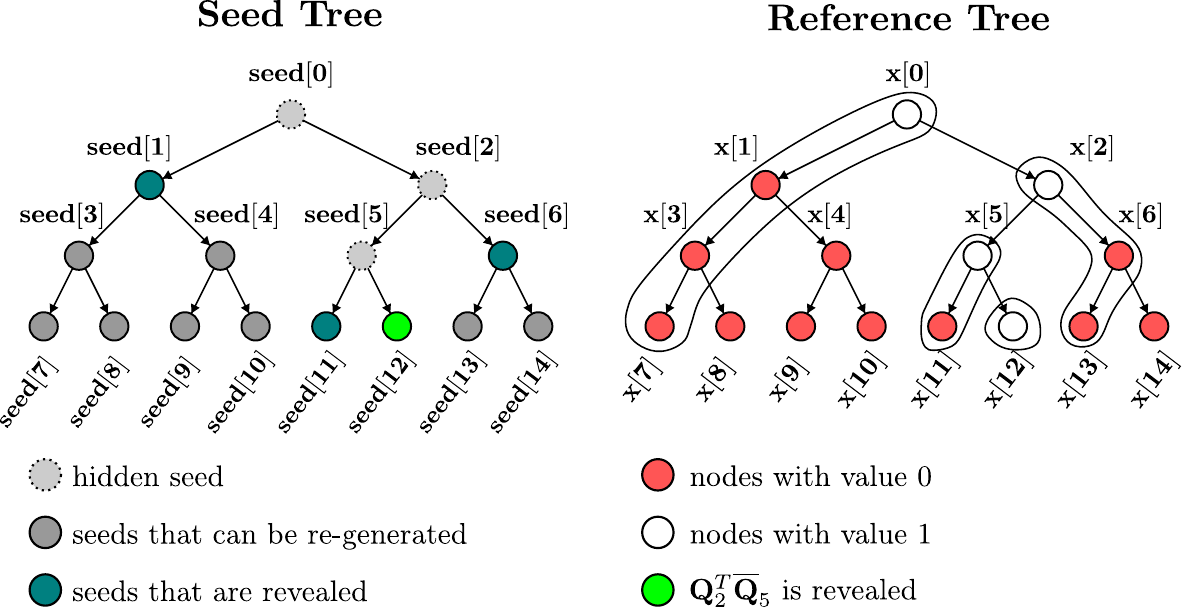}
\caption{Example for extraction of response $\vect{rsp}${$,\ \vect{TreeNode}$} according to Alg.~\ref{alg:Countermeasure2}} 
\label{fig:Counter2}  
\end{figure}
Here, we introduce another countermeasure that will keep the signature size the same as the submitted version of LESS~\cite{LESS_Specification_Doc}. In this countermeasure, our main target is that after computation of the \textit{Reference Tree} $\vect{x}$,
we check the \textit{Reference Tree} only once to compute the response for the signature. 
Note that, according to the construction of the \textit{Reference Tree} $\vect{x}$, the path from a leaf node to the root node should be of form $0^{y}1^{z}$ because if the value of any node on this path is $1$, then all the ancestors of that node will be $1$.

After the preparation of the \textit{Reference Tree} $\vect{x}$, we modify the signature generation method. First, observe that for any leaf node $\vecentry{x}{2l-1+i}$ of the \textit{Reference Tree}, let the path to the root be $\vecentry{x}{i_{0}}\vecentry{x}{i_{1}}\cdots\vecentry{x}{i_{p}}$, where $p=\lceil\log_2(l)\rceil+1$ is the height of the \textit{Reference Tree}. Here, $i_{0}=2l-1+i$ and $\vecentry{x}{i_{j}}$ is the ancestor of $\vecentry{x}{2l-1+i}$ at the height {$j$}, this means that {$\vecentry{x}{i_{p}}$ is the root}. The following is the signature generation process: 
\begin{itemize}
    \item  {\textbf{Step-1:}We start from the leftmost leaf node.}
    \item  {\textbf{Step-2:} check the path $\vecentry{x}{i_0}\vecentry{x}{i_1}\cdots\vecentry{x}{i_p}$}
    \item  {\textbf{Step-3:} If $\vecentry{x}{i_j}=1$ for $0\leq j\leq p$, then we store } {$\mat{Q}_{\vecentry{d}{i}}^{T}\overline{\mat{Q}}_{i}$ in $\vect{rsp}$} {and select the next leaf node as} {$\vecentry{x}{2l+i}$} {and goto \textbf{Step-2}, else go to \textbf{Step-4}.} 
    \item  {\textbf{Step-4:} We find $\vecentry{x}{i_h}$, which is the highest ancestor of $\vecentry{x}{i_0}$ with the value zero.}
    \item  {\textbf{Step-5:} Since $\vecentry{x}{i_{h}}=0$, all the $2^{h}$ leaf nodes of the subtree rooted $\vecentry{x}{i_{h}}$ must be zero. We store the seed $\vecentry{seed}{i_{h}}$ in $\vect{TreeNode}$ and we select the next leaf node as $\vecentry{x}{2l-1+i+2^{h}}$ and go to \textbf{Step-2}. If no more leaf nodes are left, then we stop.}
    \item {\textbf{Step-6:} return the pair $\vect{rsp}$ and $\vect{TreeNode}$.}
\end{itemize}

The digest $\vect{d}$ and the \textit{Reference Tree} are prepared in the same process here as it is prepared in Alg.~\ref{alg:Signature}. Only the vectors $\vect{rsp}$, $\vect{TreeNode}$ are prepared using Alg.~\ref{alg:Countermeasure2}. At the end, $(cmt,\ salt,\ \vect{rsp}$,\ {$, \vect{TreeNode})$} is generated as the signature. 

\begin{algorithm}
\caption{$\texttt{LESS\_Gen\_rsp\_update}$}\label{alg:Countermeasure2}
\begin{algorithmic}[1]
 \Require
    The fixed weight digest vector $\vect{d}$,
    The secret monomial matrices $\mat{Q}_{i},~\forall i\in\ring{Z}{s}$, where $\mat{Q}_{0}=\mat{I}_{n}$,
    The \textit{Seed Tree} $\vect{seed}$, 
    and the partial monomial matrices $\overline{\mat{Q}}_{j}$, $\forall j\in\ring{Z}{t}$.
 
 \Ensure The response $\vect{rsp}$ {and} {$\vect{TreeNode}$}
\For{$i=0;\ i<t;\ i=i+1$}
\If{$\vecentry{d}{i}=0$}
    \State $\vecentry{f}{i}=0$
\Else
\State $\vecentry{f}{i}=1$
\EndIf
\EndFor
\For{$i=0;\ i<4l;\ i=i+1$}
    \State $\vecentry{x}{i}=0$
\EndFor
\State $\vect{x}\xleftarrow{}\texttt{compute\_seeds\_to\_publish}(\vect{f},\ \vect{x})$

\State $i=0,$ {$j=0,\ j'=0$}
\While{$i < t$}
    \State $c=2l-1+i,~h=0,~h'=0$
    \While{{$Parent(c)\neq 0$}}
            \If{$\vecentry{x}{c}=0$}
            \State $c' = c$, {$h'=h+1$}
            \EndIf
            $c=Parent(c)$, {$h=h+1$}
        \EndWhile
        \If{{$h'=0$}}
        \State$\vecentry{rsp}{j'}=\texttt{CompressMono}(\mat{Q}_{\vecentry{d}{i}}^{T}\overline{\mat{Q}}_{i})$
        \State $i=i+1,\ j'=j'+1$
        \Else
        \State {$\vecentry{TreeNode}{j}=\vecentry{seed}{c'}$}
        \State $i=i+2^{h'-1},\ j=j+1$
        \EndIf
    \EndWhile
\State Return $\vect{rsp}${$,\ \vect{TreeNode}$}
\end{algorithmic}
\end{algorithm}

\begin{example}
   Given a fixed signature digest vector represented as $\vect{d}=(0,\ 0,\ 0,\ 0,\ 0,\ 2,\ 0,\ 0)$. First, we construct the \textit{Reference Tree} $\vect{x}$, which is illustrated in Fig.~\ref{fig:Counter2}. Begin by checking leaf nodes from the left side.
First, we take the leftmost leaf node $\vecentry{x}{7}$. The path from $\vecentry{x}{7}$ to root is {$\vecentry{x}{i_{0}}\vecentry{x}{i_{1}}\vecentry{x}{i_{2}}\vecentry{x}{i_{3}}$}=$\vecentry{x}{7}\vecentry{x}{3}\vecentry{x}{1}\vecentry{x}{0}$, and it is valued $0001$. In Fig.~\ref{fig:Counter2}, we can see that the height of the last ancestor valued $0$ is $h'=3$, and the node is $\vecentry{x}{1}$. 
We store $\vecentry{seed}{1}$ in response {$\vect{TreeNode}$} and select the next leaf node as {$\vecentry{x}{7+2^{h'-1}}=$} $\vecentry{x}{11}$. 
Final response will be calculated as {$\vect{rsp}=(\mat{Q}_{2}^{T}\overline{\mat{Q}}_{12})$} and {$\vect{TreeNode}=(\vecentry{seed}{1},\ \vecentry{seed}{11},\ \vecentry{seed}{6})$}.
\end{example}

Suppose we inject a fault at the node $\vecentry{x}{i}$ and alter its value from $1$ to $0$. Then some of its ancestors may change. Let $\vecentry{x}{i_1}$, $\vecentry{x}{i_2}$, $\cdots,$ $\vecentry{x}{i_h}$ be the list of all ancestors of $\vecentry{x}{i}$, where $\vecentry{x}{i_j}$ is ancestor of $\vecentry{x}{i_{j-1}}$ for all {$j\in[2,\ h]$} and $\vecentry{x}{i_h}$ is the root. Suppose $\vecentry{x}{i_y}$ is the highest ancestor in the list to have the value zero. Now, consider the leftmost leaf node $\vecentry{x}{r}$ of the subtree rooted at $\vecentry{x}{i_y}$, then $\vecentry{x}{i_y}$ is the highest node with value zero in the path from $\vecentry{x}{r}$ to root. Hence, according to Alg.~\ref{alg:Countermeasure2}, $\vecentry{seed}{i_y}$ is appended to {$\vect{TreeNode}$} and all the leaf nodes in the subtree rooted at $\vecentry{x}{i_y}$ are skipped.

Observe that the fault at $\vecentry{x}{i}$ only affects the subtree rooted at $\vecentry{x}{i_y}$, the rest of the \textit{Reference Tree} is unchanged. The subtree rooted at $\vecentry{x}{i_y}$ is skipped after revealing $\vecentry{seed}{i_y}$, and $\vecentry{seed}{i_y}$ can only be used to generate the ephemeral seeds that do not have any information about the secret monomial matrices. Therefore, the attack will not be possible with just one fault.

We only change the attack surface part to protect the LESS scheme against our attack. The attack surface of the CROSS signature scheme is similar to LESS. {We can use the proposed countermeasure for CROSS also. We only need to modify the update method of $\vect{rsp}$ and $\vect{TreeNode}$ according to the CROSS signing algorithm.}

\subsubsection{{Cost of the countermeasure }}
{Here we will compare the cost analysis of our proposed countermeasure with the original LESS implementation (Alg.~{\ref{alg:Signature}}). The \textit{Reference Tree} generation process in our proposed method is the same as the original LESS proposal. We have only changed the $\vect{TreeNode}$ and $\vect{rsp}$ generation process but result is same in both cases $i.e$, for a particular \textit{Seed Tree}, \textit{Reference Tree} pair, our method and original LESS implementation, both generate the same $\vect{TreeNode}$ and $\vect{rsp}$.
}
{
First of all, we consider the following computation costs:}
\begin{itemize}
    \item {$C_{\text{check}}$: cost of any condition checking}
    \item {$C_{\text{mono}}$: cost of computation of a monomial multiplication followed by a \texttt{CompressMono} function and storing the result}
    \item {$C_{\text{seed}}$: cost of storing seeds from \textit{Seed Tree} condition checking}
\end{itemize}
{Also, we fix a \textit{Seed Tree} and a \textit{Reference Tree} and we assume that $r$ many seeds from \textit{Seed Tree} are to be stored in $\vect{TreeNode}$. Assume that the total number of nodes in the \textit{Reference Tree} is $N$.}
\paragraph{Cost of LESS original implementation (Alg.~{\ref{alg:Signature}}):} 
{
As we can see in Alg.~{\ref{alg:Signature}}, the $\vect{TreeNode}$ is generated using Alg.~{\ref{alg:SeedTreePaths}}. We are going to ignore the cost of \texttt{compute\_seeds\_to\_publish} function (Alg.~{\ref{alg:ComputeSeedToPublish}}), as it has also been used in our countermeasure. For each node in the \textit{Reference Tree}, Alg.~{\ref{alg:SeedTreePaths}} checks the node and its parent node which takes $2N\cdot C_{\text{check}}$ computations. As we have assumed earlier there are $r$ many seeds which are to be stored in $\vect{TreeNode}$, which takes $r\cdot C_{\text{seed}}$ computations. After that, Alg.~{\ref{alg:Signature}} checks each value of the vector $\vect{d}$ which takes $t\cdot C_{check}$ computations and computes the monomial multiplication and calls the \texttt{CompressMono} function for each $\vecentry{d}{i}\neq 0$, which takes $w\cdot C_{\text{mono}}$ computations. Therefore the total cost is $(2N + t)\cdot C_{\text{check}} + r\cdot C_{\text{seed}} +w\cdot C_{\text{mono}}$.
}
\paragraph{Cost of our proposed method (Alg.~{\ref{alg:Countermeasure2}}):}
{
In each iteration of the while loop, we first check the full path from the leaf node to the root, this takes $\log_2{N}\cdot C_{\text{check}}$ computations. In each iteration, the algorithm can either update $\vect{TreeNode}$ or update $\vect{rsp}$ $i.e.$, the total number of iterations in the while loop is $(r+w)$. The condition checking takes total $(r+w)\cdot \log_2{N}\cdot C_{\text{check}}$ computations. Now the $\vect{rsp}$ is updated for $w$ many iterations, which takes total $w\cdot C_{\text{mono}}$ computations. Similarly, updating $\vect{TreeNode}$ takes $r\cdot C_{\text{seed}}$ computations. Therefore, the total cost is $(r+w)\cdot\log_2{N}\cdot C_{\text{check}} + r\cdot C_{\text{seed}} + w\cdot C_{\text{mono}}$.\\\\
}

{Observe that the cost of the countermeasure may vary with the value of $r$, which is the number of seeds published from \textit{Seed Tree}. We have benchmarked the performance of the LESS-sign algorithm with and without our countermeasure for all parameter sets of LESS in Table~{\ref{tab:benchmark}}. As we can see, including our countermeasure does not degrade the performance of the LESS-sign algorithm.
}
\begin{table}[!ht]
\centering
\begin{tabular}{cccc}
\hline
\multirow{2}{*}{\begin{tabular}[x]{@{}c@{}}~\bf Security\\\bf Level\end{tabular}~}&
\multirow{2}{*}{\begin{tabular}[x]{@{}c@{}}\bf Parameter\\\bf Set\end{tabular}~}&\multicolumn{2}{c}{\begin{tabular}[x]{@{}c@{}}\bf Average cpucycles\\\bf ($\times 10^{6}$ cycles)\end{tabular}} \\
\cline{3-4}
&&\begin{tabular}[x]{@{}c@{}}\bf Our\\\bf Countermeasure\end{tabular}~&
\begin{tabular}[x]{@{}c@{}}\bf Original\\\bf LESS\end{tabular}\\ \hline
\multirow{3}{*}{1}
                    & LESS-1b & 1162.31& 1162.19 \\
                   & LESS-1i &1148.03  & 1147.94 \\
                   & LESS-1s & 931.57&931.72\\ 
                   \hline
\multirow{2}{*}{3} & LESS-3b & 9563.01&9564.23 \\
                   & LESS-3s &11285.14& 11283.65 \\
                    \hline
\multirow{2}{*}{5} & LESS-5b & 44031.11&44036.56\\
                   & LESS-5s &29544.22& 29542.31\\
                    \hline
\end{tabular}%
\caption{{LESS-sign performance comparison with our countermeasure against original LESS implementation}}.
\label{tab:benchmark}
\end{table}

{
For benchmarking, we have used an HP Elite Tower 600 G9 Desktop with an Intel Core i7-12700 CPU running at 2.1 GHz and 32 GB physical memory, which was running Ubuntu 22.04.4 LTS. The test codes were executed on a single core with Turbo Boost and hyperthreading disabled. 
}

\section{Discussion and future direction}\label{sec:DiscussionFutureWork}
In this study, we have assumed a single-fault model where an attacker can only inject a fault in one single location. The countermeasure we have provided is based on that assumption. We emphasize the necessity of future investigations into higher-order fault models, side-channel attacks using power, electromagnetic radiation~\cite{DBLP:journals/iacr/KunduCSKMV23,Fault_Attacks_on_CCA-secure_Lattice_KEMs}, and combined (side-channel assisted fault attack) attack. This study is the first research study enhancing the security of the digital signature scheme LESS and CROSS against a broader spectrum of fault attacks. {LESS has $(s-1)$ secret monomial matrices, and we've shown that one pair can recover some information about one secret matrix. So, we need multiple targeted pairs to retrieve all secret matrices. This number of required pairs depends on various parameters. Therefore, we require more than one effective faulted signature for some parameter sets of LESS. For CROSS, there’s only one secret $\vect{e}$ for all the parameter sets. It can be recovered with just one targeted pair.
}

{In this work, we have done a fault analysis of the LESS{~\cite{LESS_Specification_Doc}} signature scheme that has been submitted to NIST.} However, the authors of LESS have updated the scheme in the LESS project's site~\cite{LESSProjectSite}. We observe that our mentioned attack surface, $i.e.$, the computation of $\vect{TreeNode}$ by using the function \texttt{SeedTreePaths}, are present there too. So, our attack is still applicable to their updated version. Another code-based signature scheme MEDS (Matrix Equivalence Digital Signature)~\cite{Take_your_MEDS:_Digital_Signatures_from_Matrix_Code_Equivalence} based on the zero-knowledge protocol. {Like LESS and CROSS, the Sign algorithm of MEDS uses a similar tree construction to reduce the signature size. In this case, the response $(\tilde{\mat{A}}_{i}$ or $\mat{Q}\cdot\tilde{\mat{A}}_{i})$ is constructed depending on some fixed weight digest vector $\vect{d}$, where $\mat{Q}$ is a secret component. It involves the same seed tree and \textit{Reference Tree} to store some seeds corresponding to the response $\tilde{\mat{A}}_{i}$ in a similar manner. So, the same attack model can also be applied to the MEDS signature scheme. However, we have not completely analyzed how many faulted signatures are needed to find the entire secret. We left this part for future work.} 

We have shown a fault detection method where we have fixed a position $\vecentry{x}{i}$ of the \texttt{Reference Tree} and injected fault at that location. The detection method in Section~\ref{sec:Faultdetection} can detect a successful and effective fault at the location $\vecentry{x}{i}$ for any chosen $i$, where $1\leq i< 4l-1$. Moreover, this method can determine the occurrence of an effective fault at any arbitrary location within the reference tree by applying the detection procedure for each $1\leq i< 4l-1$. Given that $l\in\{128,\ 512,\ 1024\}$ (according to Table~\ref{tab:parameters}), this approach is computationally feasible. However, a mathematical analysis for this scenario has not been included and is left for future work.
\section*{Acknowledgements}
This work was partially supported by Horizon 2020 ERC Advanced Grant (101020005 Belfort), CyberSecurity Research Flanders with reference number VR20192203, BE QCI: Belgian-QCI (3E230370) (see beqci.eu), Intel Corporation, Secure Implementation of Post-Quantum Cryptosystems (SECPQC) DST-India and BELSPO. 
Angshuman Karmakar is funded by FWO (Research Foundation – Flanders) as a junior post-doctoral fellow (contract number 203056 / 1241722N LV). Puja Mondal is supported by C3iHub, IIT Kanpur. 

%
%
%
\newpage
 \bibliographystyle{splncs04}
 \bibliography{mybibliography}
\newpage
\appendix
\section*{Supplementary material}
\section{Comparison of LESS with other code-based signature schemes}\label{appendix:comparison}
Table~\ref{tab:comparison_of_sizes} compares the key sizes and performance of LESS with other code-based digital signature schemes submitted to NIST's additional call for digital signatures~\cite{nist_additional_call}.
\begin{table}[!ht]
\resizebox{.75\columnwidth}{!}{%
\begin{tabular}{clllll}
\hline
\multirow{2}{*}{\textbf{Category}} &
  \multirow{2}{*}{\textbf{Scheme}} &
  \multicolumn{2}{c}{\textbf{\begin{tabular}[c]{@{}c@{}}Performance\\ (M Cycle)\end{tabular}}} &
  \multicolumn{2}{c}{\textbf{\begin{tabular}[c]{@{}c@{}}Size\\ (Bytes)\end{tabular}}} \\ \cline{3-6} 
                           &       & Sign.   & Verify  & Signature   & \begin{tabular}[c]{@{}c@{}}Public\\ Key\end{tabular} \\ \hline
\multirow{4}{*}{Level I}   & WAVE~\cite{WAVE_Specification_Doc}  & 1161    & 205.9   & 822     & 3677390                                              \\
                           & MEDS~\cite{Take_your_MEDS:_Digital_Signatures_from_Matrix_Code_Equivalence}  & 518.1   & 515.6   & 9896    & 9923                                                 \\
                           & CROSS~\cite{CROSS_Specification_Doc} & 22      & 10.3    & 10304   & 61                                                   \\
                           & LESS~\cite{LESS_Specification_Doc}  & 263.6   & 271.4   & {5325}  & {98202}                                              \\ \hline
\multirow{4}{*}{Level III} & WAVE~\cite{WAVE_Specification_Doc}  & 3507    & 464.1   & 1249    & 7867598                                              \\
                           & MEDS~\cite{Take_your_MEDS:_Digital_Signatures_from_Matrix_Code_Equivalence}  & 1467    & 1462    & 41080   & 41711                                                \\
                           & CROSS~\cite{CROSS_Specification_Doc} & 46.5    & 18.3    & 23407   & 91                                                   \\
                           & LESS~\cite{LESS_Specification_Doc}  & 2446.9  & 2521.4  & {14438} & {70554}                                              \\ \hline
\multirow{4}{*}{Level V}   & WAVE~\cite{WAVE_Specification_Doc}  & 7397    & 813.3   & 1644    & 13632308                                             \\
                           & MEDS~\cite{Take_your_MEDS:_Digital_Signatures_from_Matrix_Code_Equivalence}  & 1629.8  & 1612.6  & 132528  & 134180                                               \\
                           & CROSS~\cite{CROSS_Specification_Doc} & 74.8    & 26.1    & 43373   & 121                                                  \\
                           & LESS~\cite{LESS_Specification_Doc}  & 10212.6 & 10458.8 & {26726} & 132096                                               \\ \hline
\end{tabular}%
}
\caption{Comparison of code-based signature schemes in terms of performance and size.}
\label{tab:comparison_of_sizes}
\end{table}
\section{Verification algorithm of LESS}\label{sec:LESS_verification}
The verification algorithm in Alg.~\ref{alg:vrfy} takes a message $\textit{m}$ and signature $\tau$ and the public key $\text{PK}$ as inputs and returns \texttt{1}, if the $\tau$ is a valid signature of the message $m$ otherwise, it will return \texttt{0}. 


\begin{algorithm}[H]
\caption{\texttt{LESS\_Vrfy}(${m},\ \tau,\ \text{PK}$)}\label{alg:vrfy}
\begin{algorithmic}[1]
\Require A message $\textit{m}$, the public key \text{PK} and the signature $\tau=(\textit{salt},\ \textit{cmt},\ \vect{TreeNode},\ \vect{rsp})$.
\Ensure It will return $\texttt{1}$, if $(\textit{m},\ \tau)$ is a valid message signature pair. Otherwise, it will return $\texttt{0}$.
\State $\vect{d'}\leftarrow \texttt{CSPRNG}(\textit{cmt},\ \mathbb{S}_{t,w})$
\For{$i=0;\ i<t;\ i=i+1$}
\If{$\vecentry{d'}{i}= 0$}
    \State $\vecentry{f'}{i}=0$ 
\Else
    \State $\vecentry{f'}{i}=1$
\EndIf
\EndFor
\State $ \vect{ESEED'}\leftarrow \texttt{regenerate\_leaves}(\textit{salt},\ \vect{TreeNode},\ \vect{f'})$
\State $\mat{G}_{0}\leftarrow \texttt{CSPRNG}(\textit{gseed},\ \mathbb{S}_\texttt{RREF})$
\State $k=0$
\For{$i=1;\ i<t;\ i=i+1$}
    \If{$\mathbf{d'}[i]=0$}
        \State $\widetilde{\mat{Q}}'_{i}\xleftarrow{}\texttt{CSPRNG}(\vecentry{ESEED'}{i},\ \mat{M}_{n})$
        \State $(\overline{\mat{Q}}'_{i},\ \overline{\mat{V}}'_{i})\xleftarrow{} \texttt{PreparedDigestInput}(\mat{G}_{0},\ \widetilde{\mat{Q}}'_{i})$ 
    \Else
        \State $j=\vecentry{d'}{i}$
        \State $\mat{G}_{j}\leftarrow \texttt{ExpandRREF}(\text{PK}[j])$
        \State $\mat{Q}^*\leftarrow \texttt{ExpandToMonomAction}(\vecentry{rsp}{k})$
        \State Compute $J\leftarrow \left\{\alpha_{i}:\mat{Q}^*[\alpha_{i}, *]=0\right\}$
        \State $\widehat{\mat{G}}\leftarrow \left(\mat{G}_{j}\mat{Q}^*~|~ \mat{G}_{j}[*,\ J]\right)$
        
        \State ($\Hatt{\mat{G}},\ pivot\_column)\leftarrow \texttt{RREF}(\Hatt{\mat{G}})$
        \State $\textit{NP}=0,\ \overline{\mat{V}}'_{i}=\mat{O}$
        \For{ $c=0;\ c<n;\ c=c+1$} 
            \If{$pivot\_column[c]=0$} 
                \State $\overline{\mat{V}}'_{i}\leftarrow \texttt{LexMin}(\widehat{G},\ \overline{\mat{V}}'_{i},\ \textit{NP},\ c)$
               \State $\textit{NP}=\textit{NP}+1$
            \EndIf
        \EndFor
        \State $\overline{\mat{V}}'_{i}\leftarrow\texttt{LexSortCol}(\overline{\mat{V}}'_{i})$
        \State $k=k+1$
    \EndIf

\EndFor
\State $\textit{cmt}'\leftarrow \texttt{H}(\overline{\mat{V}}'_{0},\ \cdots,\ \overline{\mat{V}}'_{t-1},\ \textit{m},\ \textit{len},\ \textit{salt})$
\If{$\textit{cmt}=\textit{cmt}'$}
\State Return $\texttt{1}$
\Else
\State Return $\texttt{0}$
\EndIf
\end{algorithmic}
\end{algorithm}

\end{document}